\documentclass[11pt]{article}
\usepackage{graphicx,amsmath,amsfonts,bbm,mathtools,changepage,soul}
\usepackage{algorithm,algorithmic,caption,comment,amssymb,hyperref}
\usepackage[usenames,dvipsnames]{color}
\usepackage[round]{natbib}
\usepackage[height=9in,a4paper,hmargin={3cm,3cm}]{geometry}
\usepackage{natbib}

\newcommand\blfootnote[1]{%
  \begingroup
  \renewcommand\thefootnote{}\footnote{#1}%
  \addtocounter{footnote}{-1}%
  \endgroup
}

\definecolor{darkred}{RGB}{150,0,0}
\definecolor{darkgreen}{RGB}{0,150,0}
\definecolor{darkblue}{RGB}{0,0,200}
\hypersetup{colorlinks=true, linkcolor=darkred, citecolor=darkgreen, urlcolor=darkblue}

\numberwithin{equation}{section}
\def \endprf{\hfill {\vrule height6pt width6pt depth0pt}\medskip}
\newenvironment{proof}{\noindent {\bf Proof} }{\endprf\par}

\newtheorem{theorem}{\textbf{Theorem}}
\newtheorem{lemma}{\textbf{Lemma}}
\newtheorem{corollary}{\textbf{Corollary}}
\newtheorem{example}{\textbf{Example}}
\newtheorem{definition}{\textbf{Definition}}

\newtheorem{remark}{\textbf{Remark}}

\newtheorem{proposition}{\textbf{Proposition}}
\newtheorem{assumption}{\textbf{Assumption}}

\input{tcilatex}
\begin{document}

\title{Optimal dynamic information provision in traffic routing\blfootnote{Research partially supported by the SNSF grant P300P2 177805 and by ARO MURI W911NF1810407. The authors would like to thank A. Makhdoumi and A. Malekian for useful discussions.}}
\author{Emily Meigs\footnote{Operations Research Center, Massachusetts Institute of Technology, Cambridge, MA, 02139 emeigs@mit.edu
}, Francesca Parise\footnote{Department of Electrical Engineering and Computer Science, Massachusetts Institute of Technology, Cambridge,
MA, 02139 parisef@mit.edu}, Asuman Ozdaglar\footnote{Department of Electrical Engineering and Computer Science, Massachusetts Institute of Technology, Cambridge,
MA, 02139 asuman@mit.edu}, Daron Acemoglu\footnote{Department of Economics, Massachusetts Institute of Technology, Cambridge,
MA, 02139 daron@mit.edu}}
\maketitle

\begin{abstract}
We consider a two-road dynamic routing game where the state of one of the
roads (the \textquotedblleft risky road\textquotedblright ) is stochastic
and may change over time. This generates room for experimentation. A central
planner may wish to induce some of the (finite number of atomic) agents to
use the risky road even when the expected cost of travel there is high in
order to obtain accurate information about the state of the road. Since
agents are strategic, we show that in order to generate incentives for
experimentation the central planner however needs to limit the number of
agents using the risky road when the expected cost of travel on the risky
road is low. In particular, because of congestion, too much use of the risky
road when the state is favorable would make experimentation no longer
incentive compatible. We characterize the optimal incentive compatible
recommendation system, first in a two-stage game and then in an
infinite-horizon setting. In both cases, this system induces only partial,
rather than full, information sharing among the agents (otherwise there
would be too much exploitation of the risky road when costs there are low).
\end{abstract}

\section{Introduction}

Road conditions and congestion, which determine drivers' travel times, can
change unpredictably due to accidents, traffic jams, construction, and
weather conditions. GPS technology and smartphone routing applications such
as Waze and Google Maps promise to provide information to drivers (agents)
about real-time road conditions and congestion levels to reduce travel
times. These technologies, however, depend on real-time accurate information
from drivers on the road. This introduces an exploration-exploitation
trade-off: the central planner (CP) or the routing applications need to send
some drivers to roads with unknown and unfavorable conditions to obtain
up-to-date information. In contrast to standard experimentation models, the
choices of the CP are constrained by incentive compatibility of the agents:
a driver will only experiment if he (or she) expects to later benefit from
the information he obtains now. This will typically necessitate other
drivers not to have full information, since otherwise they would also
equally benefit by using the better road and in the process cause congestion
for the experimenter.

We investigate these issues by developing a two-road dynamic routing model
with congestion and a finite number of (atomic) forward-looking agents
(drivers). We analyze this model first in a two-stage setting and then turn
to an infinite-horizon setting where the condition of the risky road changes
according to a two-state Markov chain. In this environment we characterize
the optimal incentive compatible recommendation system for the CP. Incentive
compatibility requires that experimenters be rewarded, which means that the
roads they will be using in the future should not be too congested. This is
the reason why the social planner cannot induce full information among
drivers (for example, by sharing the exact road conditions), since this
would make non-experimenters also take the best roads, leading to excessive
congestion for the experimenters. There is another aspect to incentive
compatibility, however. Non-experimenters need to receive sufficient utility
from following the recommendations of the CP, otherwise they may deviate
from the recommendations, increasing congestion and also reducing the
rewards to experimentation. These two aspects tightly bound the number of
agents that the CP can send to the risky road after seeing favorable
conditions, to both guarantee that experimenters are rewarded and
non-experimenters would not want to deviate.

More formally, we assume that for the \textit{safe} road travel time depends
on the flow via a known affine function. For the \textit{risky} road travel
time is a linear function of the flow with unknown stochastic congestion
coefficient, $\theta $. Each agent minimizes his discounted travel time,
while the CP minimizes the sum of all discounted travel times. When there is
experimentation (meaning that some agent is using the risky road) the CP
learns the state of the road (i.e. the congestion coefficient) at that
time, and then makes recommendations to all agents on the basis of this
information. Agents themselves learn the state of the risky road if they
take it; otherwise they rely on the recommendations sent by the CP and their
observations of the flows to form beliefs.

We first fully characterize the optimal incentive compatible recommendation
system in a two-stage game where the state of the risky road $\theta $ is
either low (L) or high (H), and remains constant over time. This state $%
\theta $ is unknown to both the CP and all the agents at the beginning of
the first stage. Under full information all agents will learn $\theta $
before the second stage (as long as at least one agent takes the risky road
in the first stage), and we show that this leads to insufficient incentives
for experimentation. We then establish that pure private information (where
only the experimenters know the state of the risky road) generates lower
average travel times than full information, under some priors. This is
because private information encourages experimentation. We then characterize
the optimal recommendation system, which induces partial information,
meaning that the CP recommends the use of the risky road when its conditions
are good to some of the non-experimenters. This is beneficial for reducing
their travel time, while still maintaining the incentive compatibility of
experimentation.

We then extend our model to an \textit{infinite-horizon setting}. In this
case we assume the state of the risky road changes at each stage of the game
according to an underlying Markov chain ($\theta _{t}\in \{L,H\}$). This
framework could model potential {lane closures} that persist over time, for
example due to construction, or multi-day events. We show that our general
insights from the two-stage model generalize to this more complex setting. A
new challenge in this case is that the CP must account for the knowledge
that users infer based on their observations of others actions. For example,
an agent on safe that sees many agents switching to risky can infer that
that risky road was low at the previous round and will therefore be low
again with high probability (depending on the underlying transition
probabilities). The CP must consider the agents' ability to infer
information with one step delay and mitigate the possible deviations agents
may take to improve their cost. We characterize an incentive compatible
recommendation scheme that leads to better social cost than full information
and we study optimality of such scheme under different assumptions. We again
find that the CP must balance the amount of information it shares with
agents to account for both efficiency and incentive compatibility.

\subsection{Related Literature}

A classic example of exploration-exploitation trade off is given by the
multi-armed bandit (MAB) problem. The seminal analysis of \cite{gittins1979bandit} focuses on the case of a single forward-looking agent
that can choose between different arms with unknown reward distributions.
This introduces an experimentation-exploitation trade-off for the agent as
he decides which arm to pull in order both to see rewards and improve his information for the future.

A large literature has been devoted to extend the MAB model to different
settings. Among these works, the most closely related to ours are the works
of \cite{bolton1999strategic} and \cite{keller2005strategic}, where multiple
experimenters can learn from one another. These papers show that because of
free riding by agents there will be less experimentation than in the
standard MAB model since agents can learn from the effort of others. Our
setting is different because of three features: first, congestion creates
dependent payoffs across agents; second, we focus on the incentive
compatible experimentation scheme for a central planner; and third,
information in our setting is neither public (as in \cite{bolton1999strategic} and \cite{keller2005strategic}) nor fully private,
since the recommendation system will share some of the information the CP
acquires.

Previous results have shown that, at least for static settings, providing
public information in routing games is not necessarily socially optimal, as
detailed in \cite{acemoglu2018informational} and \cite{liu2016effects}. For
example, \cite{acemoglu2018informational} show that, under certain
conditions, increasing the information of a subset of agents can make these
agents worse off. These results motivate us in considering different models
of information sharing. To this end, we introduce a central planner (CP)
that controls the provision of information between agents.

Our setting is also related to models of Bayesian persuasion, see for
example \cite{kamenica2011bayesian} and \cite{bergemann2016information}.
This framework has been applied to study \textit{static traffic problems} in 
\cite{das2017reducing, tavafoghi2017informational,lotfi2018bayesian,
wu2019information}, social information sharing in \cite{li2017dynamic} and
competition among multiple information providers in \cite{varaiya2019platform}. The key difference between these works and our paper
is that we consider a \textit{dynamic problem} where: i) the game is dynamic
and agents, as well as the CP, are forward looking, ii) agents influence
each other payoff functions because of congestion effects, iii) the
parameters underlying the game change over time, requiring constant
experimentation and iv) the CP depends on agents to gain information through
experimentation and thus has an informational advantage only over agents
that did not experiment.

The closest works to ours are \cite{tavafoghi2017informational} and \cite{li2018sequential}. \cite{tavafoghi2017informational} consider a two-stage,
two-road routing problem. Therein however the CP has perfect information
about the state of the risky road and does not need to rely on agents
experimentation. As a consequence, \cite{tavafoghi2017informational} focus
on the case of non-atomic agents, while we consider atomic (a finite number
of) agents, which implies that agents take into account the information they
generate for their own and others' future use. In addition, observations of
traffic flow in \cite{tavafoghi2017informational} are fully revealing, while
in our setting there is a one step delay due to the fact that the CP's
suggestions are based on previous observations and not on the current state
of the risky road (which is unknown to the CP). As already noted, we also
extend our analysis beyond two stages by looking at an infinite horizon
model where the state of the road changes in time. \cite{li2018sequential}
consider a repeated game, again in a non-atomic setting where the CP does
learn from agents' actions, but agents themselves are not learning -- agents
at each time step only use the public message that is sent by the CP to make
a myopic routing decision.

On a more general note, our results contribute to a growing literature on
social learning. For example, \cite{che2015optimal} and \cite%
{che2017recommender} study how a recommender system may incentivize users to
learn collaboratively on a product, \cite{shah2018bandit} study how
correlated preferences between agents may effect learning, \cite%
{johari2016matching} study a matching problem between heterogeneous jobs and
workers, where the aim is to learn worker types, \cite{iyer2014mean} study
learning in repeated auctions, \cite{bistritz2018characterizing} study
product adoption. 
A recent survey of the literature at the interface of learning,
experimentation and information design can be found in \cite%
{horner2016learning}. The main feature distinguishing the routing problem
addressed in our work and the applications considered in the works above is
the presence of dependent payoffs because of congestion effects, which
fundamentally modify the results since agents do not only affect each other
in their learning process but also in the received payoffs.

The rest of the paper is organized as follows. In Section \ref{section:model}%
, we introduce the routing model. In Section \ref{section:2stage} we
consider a two-stage setting, compare full information, private information,
and partial information and explain how to characterize the optimal
recommendation scheme. Finally, in Section \ref{section:infty} we detail the
infinite time horizon setting and study the optimal, incentive compatible
recommendation scheme. All proofs are given in the Appendix.

\section{Model}

\label{section:model} We consider a dynamic mechanism design problem in
which a central planner (CP) aims to minimize total travel time in a
repeated routing game with $N$ (atomic) agents on two roads. Agents decide
their own routing to minimize their own travel time, and the CP can try to
influence their choices by providing information.

\subsubsection*{\textbf{Congestion model}}

We consider a network with two roads. One of the roads, the \textit{safe road%
}, has a congestion-dependent but non-stochastic affine cost $%
S_{0}+S_{1}(N-x_{R})$ where $S_{0} > 0$, $S_{1} \geq 0$ and $x_{R}$ is the
number of agents on the risky road. The other road, the \textit{risky road},
has a linear cost $\theta _{t}x_{R}$ where $\theta _{t}\in \{L,H\}$ (where $%
L,H$ are scalars with $L, H > 0$) is an unknown congestion parameter that
changes over time according to an underlying Markov chain with switching
probabilities $\gamma _{L}:=\mathbb{P}(\theta _{t}=H|\theta _{t-1}=L)$ and $%
\gamma _{H}:=\mathbb{P}(\theta _{t}=L|\theta _{t-1}=H).$\footnote{%
Our analysis can be easily generalized to a risky road with affine cost $%
\theta_t x_R + R_0$ for known $R_0 \geq 0$. We omit this for simplicity of
exposition.} The parameter $\theta _{t}=L$ represents cases where the
congestion parameter is \textquotedblleft low\textquotedblright\ which means
the risky road is favorable, alternatively $\theta _{t}=H$ means the
parameter is \textquotedblleft high\textquotedblright\ and the safe road is
preferable.

\begin{assumption}[Frequency of switching]
\label{ass:gamma} 
\begin{equation*}
0 \leq \gamma_L, \gamma_H \leq \frac{1}{2}.
\end{equation*}
\end{assumption}
Assumption \ref{ass:gamma} imposes an upper limit on the probability that
the road condition changes. Intuitively this condition suggests that it
might be worthwhile for an agent to experiment, since the road is likely to
remain in the same condition for more than one stage.

\subsubsection*{\textbf{Agent actions and stage cost function}}

At each time $t,$ each agent $i$ chooses an action $\alpha _{t}^{i}\in
\{S,R\}$ corresponding to either taking the safe ($\alpha _{t}^{i}=S$) or
the risky road ($\alpha _{t}^{i}=R$). An agent's realized stage cost is then
given by the travel time he experiences at stage $t$ 
\begin{equation*}
\tilde{g}(\alpha _{t}^{i},\alpha _{t}^{-i},\theta _{t})=%
\begin{cases}
S_{0}+S_{1}(N-x_{t}^{R}(\alpha _{t})) & \mbox{if}\quad \alpha _{t}^{i}=S, \\ 
\theta _{t}x_{t}^{R}(\alpha _{t}) & \mbox{if}\quad \alpha _{t}^{i}=R,%
\end{cases}%
\end{equation*}
where $x_{t}^{R}(\alpha _{t}):=\sum_{i=1}^{N}\mathbbm{1}\{\alpha
_{t}^{i}=R\} $ is the total flow on the risky road at time $t$.

\subsubsection*{\textbf{Information structure}}

We assume that if an agent experiments by using the risky road at time $t$,
he observes the true value of $\theta _{t}$ and this is also directly
observed by the CP (because he communicates it truthfully to the CP or
because of direct observation of his experience by the CP via GPS).
Consequently the CP knows $\theta _{t}$ if and only if at least one agent
takes the risky road at time $t$. Before $t=0$, the CP commits to a
signaling scheme to disseminate information to agents that are on the safe
road and are thus uninformed. From here on we restrict our attention to 
\textit{recommendation schemes}, which are a specific type of signaling
scheme where the signal sent to each agent is a recommendation to take
either the safe ($r_{S}$) or the risky road ($r_{R}$)  and we
refer to such a recommendation scheme as $\pi$.\footnote{%
A formal definition of recommendation schemes for the two stage model is
given in Definition \ref{def:two_stage} and for the infinite horizon model
is given in Definition \ref{def:am}.}

\subsubsection*{\textbf{Cost function}}

\label{subsect:agents_objective}

Drivers are homogeneous and risk-neutral. Each minimizes his expected sum of
travel times over $T$ repetitions of the game 
discounted in time by a factor $\delta \in \lbrack 0,1)$. Let $h_{t-1}^{i}$
be agent $i$'s history after round $t-1$ and before time $t$ 
(based on his observations and on the signals sent by the CP). 
This includes the past actions of the agent, the flows the agent
experienced, the recommendations the agent received, and the state of the
risky road for those times the agent took it. An agent's strategy $\xi
_{t}^{i}(h_{t-1}^{i})$ maps any history to an action $\{S,R\}$. 
The agent's expected cost from time $t$ to time $T$ is given by 
\begin{equation*}
u_{t}^{i}(h_{t-1}^{i})=\mathbb{E}\left[ \sum_{k=t}^{T}\delta ^{k-t}\tilde{g}%
(\xi _{k}^{i}(h_{k-1}^{i}),\xi _{k}^{-i}(h_{k-1}^{-i}),\theta _{k})\right] .
\end{equation*}%
We focus on pure strategies; hence the expectation here is on the state of
the risky road $\theta $ and on the information received from the CP.

Each agent chooses $\xi _{t}^{i}(\cdot )$ to minimize his total expected
cost given the strategies of others and the recommendation scheme the CP
uses. Note that the agent's strategy is chosen \textit{after} the CP has
committed to a recommendation scheme.

The CP's objective is to select a recommendation scheme $\pi$ to minimize
the expected total discounted travel times of agents over $T$ stages. Let $%
\xi _{t}^{i\mid \pi }$ be the strategy agent $i $ uses in equilibrium under
the recommendation scheme $\pi $.

\begin{definition}
\label{def:IC} A recommendation scheme $\pi$ is \textit{incentive compatible}
if $\xi _{t}^{i\mid \pi }=R$ for any agent $i$ that receives a
recommendation $r_{R}$ for stage $t$ and $\xi _{t}^{i\mid \pi }=S$ for any
agent $i$ that receives a recommendation $r_{S}$ for stage $t$.
\end{definition}

Let $\Pi$ be the class of incentive compatible recommendation schemes and $%
g(x^R,\theta)$ be the total cost of a stage when there are $x^R$ agents on
the risky road and the road condition is $\theta$, i.e. 
\begin{equation*}
g(x^R,\theta) = \theta (x^R)^2 + (S_0+S_1(N-x^R))(N-x^R).
\end{equation*}
If no agent is using the risky road, we define $g(0) = (S_0+S_1N) N$ as the
cost does not depend on $\theta$. Before the game begins the CP and all
agents share a common prior on the state of the risky road which we denote
by $\beta =\mathbb{P}(\theta_0 =L)$. The CP wants to choose a scheme $\pi
\in \Pi$ to minimize 
\begin{align}  \label{equation:cp_objective}
V^{\pi}_{T}(\beta):= \mathbb{E} \left[\sum_{k=1}^{T} \delta^k g({x}^R_k,
\theta_k) \mid \mathbb{P}[\theta_0=L]=\beta \right],\quad\text{where } x^R_k
= \sum_{i=1}^N \mathbbm{1} \{ \xi_k^{i\mid\pi} (h^i_{k-1}) = R \}.
\end{align}

\section{The two-stage model}

\label{section:2stage}

We start by considering a two-stage model ($T=2$) and, for simplicity, we
assume that the road condition does not change between the two stages, that
is $\theta_0 = \theta _{1}=\theta _{2}=:\theta $ (i.e. $\gamma _{H}=\gamma
_{L}=0$). We also assume no discounting ($\delta =1$) to simplify the
derived bounds. Define the expected value of $\theta $ at the beginning of
the first stage as 
\begin{equation*}
\mu _{\beta }:=\mathbb{E}[\theta ]=\beta L+(1-\beta )H.
\end{equation*}%
In this context the objective of the CP is to minimize 
\begin{equation*}
V_{2}^{\pi }(\beta ):=\mathbb{E}\left[ g(x_{1}^{R},\theta
)+g(x_{2}^{R},\theta )\mid \mathbb{P}(\theta =L)=\beta \right] .
\end{equation*}%
We adopt the following assumption.

\begin{assumption}[Two stage model]
\label{assumption:2stage} The parameters are such that

\begin{enumerate}
\item $L < S_0 + S_1$,

\item $S_0 + S_1 N < \mu_{\beta}$.
\end{enumerate}
\end{assumption}

Assumption \ref{assumption:2stage}.1 states that if the congestion parameter
is $L$, the risky road is preferable (i.e. the cost of one agent on risky if 
$\theta = L$ is less than the cost of one agent on safe). Assumption \ref%
{assumption:2stage}.2 states that the expected cost of experimentation for
one agent $\mu_{\beta}$ is greater than a fully congested safe road. Note
that agents may nonetheless select the risky road in the first round since,
if $\theta = L$, they can exploit this information in the second round. We
let $x_L^{\QTR{up}{eq}}$ be the myopic equilibrium flow on risky if all
agents know that $\theta = L$.\footnote{%
That is $x_L^{\QTR{up}{eq}}{\ \approx} \min \left\lbrace \frac{S_0 + S_1 N}{%
L+S_1}, N\right\rbrace.$ The approximation comes from the fact that $x_L^{%
\QTR{up}{eq}}$ must be an integer as we work with an atomic model (finite
number of agents).}

\subsection{Full and private information}

We start with two extreme and simple informational scenarios:

\begin{itemize}
\item \textit{Full information} - if any agent takes the risky road in round
one, then all agents learn $\theta$ before round two.

\item \textit{Private information} - any agent that takes the risky road in
round one knows the value of $\theta$ before round two, but any agent that
chose to play safe in the first round has no new information before the
beginning of round two.
\end{itemize}

We first show that in any pure strategy Nash equilibrium, at most one agent
experiments in the first round. We then use this result to characterize the
equilibrium under both full and private information. While the result that
only one agent experiments under full information is very general, the fact
that only one agent experiments under private information is a consequence
of some of the special features of this example. In particular, in a
two-period model, there is only a limited time during which an experimenter
can exploit his information. Since we assume a linear cost function, an
additional experimenter increases the travel time sufficiently such that it
is not worthwhile for two agents to experiment. This result does not apply,
for example, in our infinite-horizon model, studied in the next section.
Nevertheless, we will see that, in that setting too, in the incentive
compatible optimal mechanism, the CP will induce only one agent to
experiment.

\begin{lemma}
{\ \label{lemma:pure_one_exp} Under Assumption \ref{assumption:2stage} and
any information scheme, in any pure strategy Nash equilibrium at most one
agent experiments in the first round.}
\end{lemma}

\begin{theorem}
\label{thrm:private_vs_full} Under Assumption \ref{assumption:2stage} the
unique pure strategy Nash equilibrium is as follows:\footnote{%
Uniqueness here refers to the total number of agents in each road.}

\begin{itemize}
\item Full information:

\begin{itemize}
\item all agents play safe in both rounds, if 
\begin{equation*}
\beta <\frac{H-(S_{0}+S_{1}N)}{H-L+(S_{0}+S_{1}N)-\frac{g(x_{L}^{\QTR{up}{eq}%
},L)}{N}}=:\beta _{f}
\end{equation*}

\item otherwise, one agent experiments in the first round, and $x_{L}^{%
\QTR{up}{eq}}$ agents use the risky road in the second round if $\theta =L$,
and all play safe if $\theta =H$.
\end{itemize}

The expected cost under equilibrium is 
\begin{equation*}
V_{2}^{full}(\beta ):=%
\begin{cases}
2g(0) & \mbox{if}\quad \beta <\beta _{f} \\ 
g(1,\mu _{\beta })+\beta g(x_{L}^{\QTR{up}{eq}},L)+(1-\beta )g(0) & \mbox{if}%
\quad \beta _{f}\leq \beta.%
\end{cases}%
\end{equation*}

\item Private information:

\begin{itemize}
\item all agents play safe in both rounds if%
\begin{equation*}
\beta <\frac{H-(S_{0}+S_{1}N)}{H+(S_{0}+S_{1}N)-2L}=:\beta _{p}\leq \beta
_{f}
\end{equation*}

\item otherwise, one agent experiments in the first round and uses the risky
road if $\theta =L$ and the safe road if $\theta =H$ in the second round.
All other agents play safe in both rounds.
\end{itemize}

The expected cost under equilibrium is 
\begin{equation*}
V_{2}^{private}(\beta ):=%
\begin{cases}
2g(0) & \mbox{if}\quad \beta <\beta _{p} \\ 
g(1,\mu _{\beta })+\beta g(1,L)+(1-\beta )g(0) & \mbox{if}\quad \beta
_{p}\leq \beta.%
\end{cases}%
\end{equation*}
\end{itemize}
\end{theorem}

\begin{corollary}
\label{corollary:private_better} If the prior belief $\beta$ is such that 
\begin{align}  \label{eq:mup_bound}
\beta_p \leq \beta < \beta_f,
\end{align}
then in the pure strategy equilibrium there is experimentation under private
information, but not under full information. Consequently private
information has a lower expected cost ($V_2^{\text{private}}(\beta) < V_2^{%
\text{full}}(\beta)$).
\end{corollary}

According to the above corollary, there may exist a range of priors where it
is better for the CP to provide no information rather than full information.
Intuitively, this happens because providing the information that $\theta =L$
to all the agents induces congestion in the second round, thus reducing the
value of information. This decreases the incentive of an agent to experiment
in the first round. In other words, full information allows more agents to
free-ride off one agent's experimentation, reducing the payoff of the
experimenter due to congestion effects. The next example illustrates the
costs as a function of the prior belief.

\begin{example}
\label{example:2stage} Suppose $N=40$, $S_{0}=10,S_{1}=1$, $L=0.9$, and $%
H=150$. The comparison of the equilibrium cost for all beliefs satisfying
Assumption~\ref{assumption:2stage} is shown in Figure \ref{figure:2stage}.
For $\beta \in \lbrack 0.50,0.57],$ there is experimentation under private
information, but no experimentation under full information.
\end{example}

The fact that full information, where the conditions of the risky road are
communicated to all agents, is not socially optimal motivates the rest of
our analysis. We will show that some amount of information sharing by the CP
is preferable to private information and characterize the optimal
recommendation scheme.

\subsection{Unconstrained social optimum}

Define $x_L^{\QTR{up}{SO}}, x_H^{\QTR{up}{SO}}$ as the social optimum
integer myopic flows on the risky road when it is known that $\theta = L$ or 
$\theta = H$ respectively, that is, 
\begin{align*}
x_L^{\QTR{up}{SO}}:= \operatornamewithlimits{argmin}_{x \in \{0, 1,...,N\}}
g(x, L), \quad x_H^{\QTR{up}{SO}} := \operatornamewithlimits{argmin}_{x \in
\{0, 1,...,N\}} g(x, H).
\end{align*}

\begin{theorem}
\label{lemma:so_2stage} Under Assumption \ref{assumption:2stage} the social
optimum is given by

\begin{itemize}
\item All agents playing safe in both rounds, if 
\begin{equation*}
\beta \leq \frac{H + g(x_H^{\QTR{up}{SO}},H)-(S_0(N+1)-S_1((2+N)N-1))}{%
H-L+g(x_H^{\QTR{up}{SO}},H)-g(x_L^{\QTR{up}{SO}},L)} := \beta_{\text{SO}}\le
\beta_p
\end{equation*}

\item One agent experimenting in the first round and $x_{L}^{\QTR{up}{SO}}$ (%
$x_H^{\QTR{up}{SO}}$) agents taking the risky road in the second round if $%
\theta=L$ ($\theta=H$), otherwise.
\end{itemize}

The expected cost under the social optimum is then 
\begin{align*}
V_2^{*} := 
\begin{cases}
2 g(0) & \mbox{if} \quad \beta < \beta_{\text{SO}}, \\ 
g(1, \mu_{\beta})+ \beta g(x_L^{\QTR{up}{SO}}, L) + (1-\beta) g(x_H^{%
\QTR{up}{SO}}, H) & \mbox{if}\quad \beta\ge \beta_{\text{SO}}.%
\end{cases}%
\end{align*}
\end{theorem}

\begin{remark}
\label{rem} Two remarks are in order. First, note that while when $\theta =H$
it is never myopically a best response for an agent to take the risky road,
the previous lemma shows that the CP may still want to send some agents to
the risky road in the second round (if $x_{H}^{\QTR{up}{SO}}\geq 1$) to
reduce congestion on safe for all other agents. Second, note that at least
for any belief $\beta \in \lbrack \beta _{\text{SO}},\beta _{p})$ the social
optimum scheme is not incentive compatible. In fact, since $\beta <\beta _{p}
$ it is not incentive compatible for an agent to experiment in the first
round (under private information the experimenter has the highest possible
gain from experimentation hence if experimentation doesn't happen under
private information it cannot happen under any information scheme).
Nonetheless, for $\beta >\beta _{SO}$ the CP would like to experiment by
sending one agent to the risky road (because knowing the state of the road
is collectively beneficial). In Example \ref{example:2stage}, $\beta _{\text{%
SO}}=0.05$ is significantly lower than $\beta _{p}=0.50$ suggesting that the
social optimum may not be incentive compatible for a large range of beliefs.
\end{remark}

\subsection{Partial information}

\begin{figure}[tbp]
\caption{ Example \protect\ref{example:2stage}. We distinguish four cases
based on the prior $\protect\beta$: A) no experimentation, B)
experimentation under social optimum, C) experimentation under private and
optimal information, D) experimentation under all schemes. }
\label{figure:2stage}\centering
\includegraphics[width=.8\textwidth]{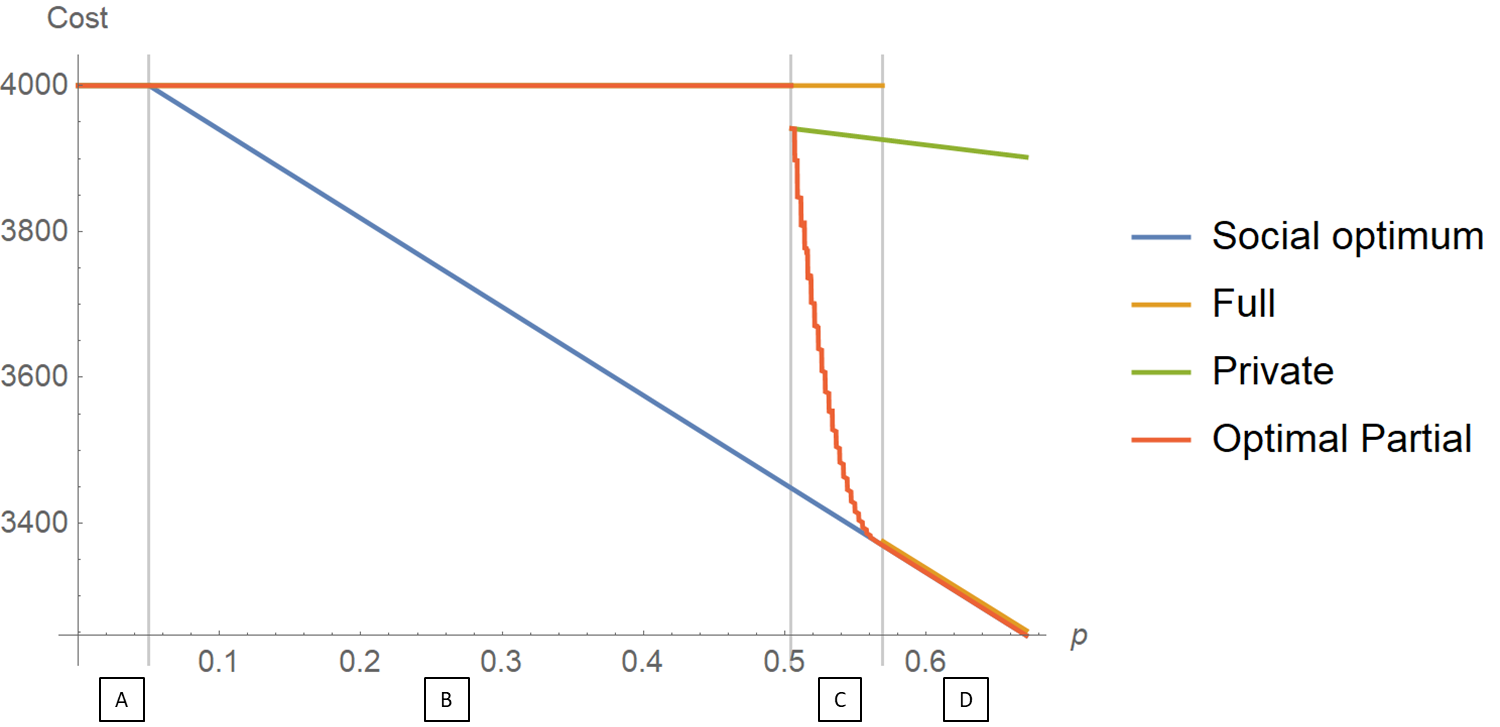}
\end{figure}

The CP can alleviate the problems of full and private information and
achieve a cost that is closer to social optimum by providing recommendations
in a coordinated way. The objective here is to find a balance between

\begin{itemize}
\item providing information to a large enough number of agents in the second
stage, so that the total cost is low when $\theta=L$;

\item providing information to a small enough number of agents in the second
stage to avoid a high level of congestion on the risky road when $\theta =L$
to encourage experimentation in the first round.
\end{itemize}

We refine the CP's recommendation scheme for the two stage model as follows.

\begin{definition}[Two stage recommendation scheme]
\label{def:two_stage} In the two stage model, a deterministic recommendation
scheme is a pair of mappings $(\pi_1, \pi_2)$ where $\pi_t: \{\beta\}
\rightarrow \{ 0, 1, ..., N\}$ maps the CP's belief on the state of the
risky road at time $t$ to the number $\pi_t(\beta)$ of uninformed agents to
whom the CP sends a recommendation of risky before time $t$.\footnote{%
We assume that the $\pi_t(\beta)$ agents to which $r_R$ is sent are chosen
uniformly at random from the set of uninformed agents.} With a slight abuse
of notation we let $\pi_t(L) := \pi_t(1)$ and $\pi_t(H) := \pi_t(0)$.
\end{definition}

Note that this definition restricts attention to recommendation systems that
are \textit{anonymous}, in the sense that the recommendations for all agents
with the same beliefs (i.e. agents that took the safe road) are drawn from
the same distribution.\footnote{%
This could also be replicated with a recommendation system that is fully
anonymous, meaning that all agents receive recommendations from the same
distribution, but those with beliefs determined from their experience of the
risky road (the experimenters) will not follow these recommendations.}
Nevertheless, the recommendation system is potentially \textquotedblleft
interim asymmetric\textquotedblright\ --- meaning that some of these agents
may receive different recommendations.\footnote{%
An alternative is to impose additionally that the scheme is interim
symmetric, but mixed. In this case, all agents would receive the same
stochastic recommendation. Because we have a finite number of agents, this
would induce additional noise in traffic flows, hence we do not focus on
this case.}

Because of Lemma \ref{lemma:pure_one_exp}, in any incentive compatible
scheme it must be $\pi _{1}(\beta )\leq 1 $. We already argued in Remark \ref%
{rem} that in any incentive compatible scheme there cannot be
experimentation if $\beta <\beta _{p}$, hence in this range it must be $\pi
_{1}(\beta )=\pi _{2}(\beta )=0$. If instead $\beta \geq \beta _{p}$, we
show that the optimal incentive compatible scheme selects $\pi _{1}(\beta )=1
$ and values of $\pi _{2}(L),\pi _{2}(H)$ obtained by solving the following
quadratic integer program with quadratic constraints (corresponding to the
incentive compatibility constraints).

\begin{theorem}
\label{thrm:2stageopt} If $\beta \geq \beta _{p}$ the optimal incentive
compatible recommendation scheme is a solution to the following minimization
problem 
\begin{subequations}
\label{eq:opt2}
\begin{align}
\min_{\pi_2(L),\pi_2(H)} \quad &g(1, \mu_{\beta}) + \beta g({x}_2^{R \mid L}, L) + (1-\beta)g({x}_2^{R \mid H}, H) \\ \text{s.t.}\quad &\underbrace{\mathbb{E}_{\theta }[S_{0}+S_{1}(N-{x}_{2}^{R\mid \theta
})\mid \textup{rec. safe}]}_{\textup{follow rec. of safe}}\leq \underbrace{%
\mathbb{E}_{\theta }[\theta ({x}_{2}^{R\mid \theta }+1)\mid \textup{rec. safe}]%
}_{\textup{deviate to risky}},  \label{2stage:ic1} \\
& \underbrace{\mathbb{E}_{\theta }[\theta {x}_{2}^{R\mid \theta }\mid \textup{%
rec. risky}]}_{\textup{follow rec. of risky}}\leq \underbrace{\mathbb{E}%
_{\theta }[S_{0}+S_{1}(N-{x}_{2}^{R\mid \theta }+1)\mid \textup{rec. risky}]}_{%
\textup{deviate to safe}},  \label{2stage:ic2} \\
& \underbrace{\mathbb{E}[\theta ]}_{\substack{ \textup{exp.'s cost} \\ \textup{%
in round 1}}}+\underbrace{\beta (L{x}_{2}^{R\mid L})+(1-\beta
)(S_{0}+S_{1}(N-{x}_{2}^{R\mid H}))}_{\substack{ \textup{experimenter's cost}
\\ \textup{in round 2}}}\leq \underbrace{2(S_{0}+S_{1}N)}_{\substack{ \textup{
all playing safe} \\ \textup{in both rounds}}} \label{2stage:ic3}\\
 &{x}_2^{R \mid L} = \pi_2(L)+1, \ {x}_2^{R \mid H} = \pi_2(H) \notag \\
 & \pi_2(L),\pi_2(H) \in \{0, 1, ..., N\}  \notag
\end{align}%
{Equations \eqref{2stage:ic1}, \eqref{2stage:ic2} and \eqref{2stage:ic3} are
given in implicit form for readability, the explicit form is provided within
the proof.}
\end{subequations}
\end{theorem}

Note, that $\pi _{2}(L)=\pi _{2}(H)=0$ is a feasible solution when $\beta
\geq \beta _{p}$ (and corresponds to private information). Moreover, if $%
\beta >\beta _{f}$, $\pi _{2}(L)=x_{L}^{\QTR{up}{eq}}-1,\pi _{2}(H)=0$ is a
feasible solution and has the same social cost as full information. Hence
private and full information have always at least weakly higher cost than
the optimal incentive compatible scheme. We show in Example \ref%
{example:2stage} that the optimal incentive compatible scheme can be
strictly better than private and full information (see region C in Figure %
\ref{figure:2stage}).

\section{The infinite-horizon model}

\label{section:infty}

We now extend our analysis to an infinite-horizon setting. {For simplicity,
we restrict our attention to the case when the safe road has a fixed cost $%
S_{0}$ (i.e. we set $S_{1}=0$), so that the cost under full information is
simply $S_{0}/(1-\delta )$.}\footnote{{If $S_{1}>0$ a similar argument can
be followed to derive an incentive compatible recommendation scheme. Proving
optimality of such a scheme is more complicated, however. The main technical
difficulty in that scenario is computing the optimal punishment for
deviations. We show below that, in the case where $S_{1}=${$0$, }the optimal
punishment is providing full information.} Finally, note that even though
the travel time on the safe road does not depend on congestion levels, we
still assume that if an agent takes safe at time $t$, he observes $x_{t}^{S}$%
.} We first characterize the social optimum scheme and give an example to
illustrate why it may not be incentive compatible. We then introduce an
incentive compatible recommendation system, prove it achieves better cost
than full information and derive conditions for optimality.

\subsection{Unconstrained social optimum}

\label{section:inf_so}We first derive the \textquotedblleft
unconstrained\textquotedblright\ social optimal, meaning that we ignore the
incentive compatibility constraints of the agents.

Suppose that the CP has a belief, $\beta_{t-1} \in [0, 1]$ about the
probability that state of the road at time $t-1$ was $L$ (we use the
convention $\beta_{t-1}=0$ if $H$ was observed and $\beta_{t-1} = 1$ if $L$
was observed). If the CP had complete control of the agents the myopically
optimal flow to send at time $t$ under a generic belief $\beta_{t-1}=\beta$
would be 
\begin{equation}  \label{xbso}
x_{\beta}^{\QTR{up}{SO}}:= \operatornamewithlimits{argmin}_{x \in \{0, ...,
N\}}\ \mathbb{E} _{\theta_t}\left[ g(x, \theta_t) \mid \beta_{t-1}=\beta %
\right].
\end{equation}
Note that the myopic flow does not depend on $t$ given the properties of
Markov processes. With a slight abuse of notation we set $x_L^{\QTR{up}{SO}%
}:=x_1^{\QTR{up}{SO}}$ and $x_H^{\QTR{up}{SO}}:=x_0^{\QTR{up}{SO}}$.

In the following, we focus on cases where the cost of experimentation is
high from a myopic standpoint. Specifically, we consider cases in which $%
x_H^{\QTR{up}{SO}}=0$, so that myopically the CP has no incentive to send
agents to the risky road after observing H in the last period. Our interest
is determining conditions under which experimentation happens when the CP is
forward looking.

\begin{assumption}[Cost of experimentation]
\label{ass:mul_muh}

Define the expected value of the congestion parameter $\theta_t$, following
an observation of $\theta_{t-1}$ as 
\begin{align*}
\mu_{L} &:= \mathbb{E} [\theta_{t} | \theta_{t-1} = L] = (1-\gamma_L) L +
\gamma_L H, \\
\mu_{H} &:= \mathbb{E} [\theta_{t} | \theta_{t-1} = H] = \gamma_H L +
(1-\gamma_H) H.
\end{align*}
We assume that $S_0>3L$ and 
\begin{align*}
\mu_L &\in [L, (1/3)S_0) \\
\mu_H &\in \left[ S_0, S_0 + \delta \gamma_H \left( \frac{S_0}{3} - \mu_L
\right) \right].
\end{align*}
\end{assumption}

Intuitively, Assumption \ref{ass:mul_muh} imposes that the expected
congestion parameter $\mu _{L}$ following an observation of $L$ is small
enough such that the CP would myopically send two or more agents after
observing $L$ (that is, $x_{L}^{\QTR{up}{SO}}\geq 2$). If this were not the
case, then the CP would send the same flow after $L$ and $H$ making the
problem uninteresting. The assumption also imposes that $\mu _{H}\geq S_{0}$%
, which implies that the CP would myopically send no agent after seeing $H$,
thus making experimentation beneficial only because of forward-looking
incentives---there would be no experimentation with myopic agents. Finally,
the upper bound on $\mu _{H}$ ensures that the forward-looking CP always
find experimentation after seeing $H$ beneficial (rather than sending all
agents to safe for one or more rounds).

\begin{proposition}[Social optimum]
\label{prop:so} Under Assumptions \ref{ass:gamma} and \ref{ass:mul_muh}, $%
x_H^{\QTR{up}{SO}}=0$ and $x_L^{\QTR{up}{SO}} \geq 2$. Let us define the
social optimum recommendation scheme as a function $\pi^{\QTR{up}{SO}}$ that
maps the belief $\beta$ that the CP has about the state of the road at time $%
t-1$ to the number of agents to send to the risky road at time $t$ to
minimize total discounted travel time, that is, 
\begin{equation}  \label{so}
\pi^{\QTR{up}{SO}}(\beta):=\operatornamewithlimits{argmin}_{x \in \{0, ...,
N\}} \mathbb{E}\left[ g(x, \theta_t) + \sum_{k\ge 1} \delta^k g\left(\pi^{%
\QTR{up}{SO}}(\beta_{t-1+k}), \theta_{t+k}\right) \mid \mathbb{P}
[\theta_{t-1} = L] =\beta \right].
\end{equation}
Then 
\begin{align}
\pi^{\QTR{up}{SO}}(\beta)= \max\{1, x^{\QTR{up}{SO}}_{\beta} \},
\end{align}
with $x^{\QTR{up}{SO}}_{\beta}$ as defined in \eqref{xbso}.
\end{proposition}

Under the social optimum recommendation scheme derived above the CP sends
one agent (the experimenter) if the state of the risky road was $H$ at the
previous step (to explore) and $x_{L}^{\QTR{up}{SO}}\geq 2$ if it was $L$
(to exploit). Hence, under $\pi^{\QTR{up}{SO}}$ the CP always knows the
state of the risky road. The next example, however, shows that this scheme
is not necessarily incentive compatible. In particular, when agents make
their own routing decisions, the CP may not be able to send $x_{L}^{\QTR{up}{%
SO}}$ agents when the state is $L$.

\begin{example}
\label{example:infty_notic} Take the extreme case where $\gamma_L = 0$. Then 
$x_L^{\QTR{up}{SO}} \approx \frac{S_0}{2 {L}}$ and $x_L^{\QTR{up}{eq}}
\approx \frac{S_0}{{L}}$.\footnote{%
Again, the approximation comes from the integer constraint of our atomic
model.} Suppose that at time $t-1$ the risky road state changes from $H$ to $%
L$. Since $\gamma_L = 0$ this will be the state of the risky road from that
point forward. According to $\pi^{\QTR{up}{SO}}$, at time $t$, the CP sends $%
x_L^{\QTR{up}{SO}}$ drivers to risky to exploit the low state. After time $t$%
, under $\pi^{\QTR{up}{SO}}$, agents that were on risky at time $t$ should
remain on risky forever and agents that were on safe at time $t$ should
remain on safe forever.

However, consider an agent on safe at time $t$. After observing the flow $%
x_L^{\QTR{up}{SO}}$ at time $t$, this agent can infer that $\theta$ has
changed to $L$. Hence at time $t+1$ he knows that

\begin{enumerate}
\item[-] if he remains on safe, as prescribed by $\pi ^{\QTR{up}{SO}}$, he
will experience a cost of $S_{0}$ for all future times;

\item[-] if he switches to risky, he will experience a cost of $\approx L(%
\frac{S_{0}}{2L}+1)=\frac{S_{0}}{2}+L$ for all future times;
\end{enumerate}

Under Assumption \ref{ass:mul_muh}, $L<\frac{S_{0}}{3}<\frac{S_{0}}{2}$.
Hence following $\pi ^{\QTR{up}{SO}}$ is not incentive compatible for the
agent.
\end{example}

In the next sections, our objective is to derive incentive compatible
recommendation schemes that achieve lower cost than providing full
information.

\subsection{Partial information: Incentive compatibility}

Example \ref{example:infty_notic} shows that the social optimum scheme $\pi
^{\QTR{up}{SO}}$ may not be incentive compatible because it does not take
into account the fact that agents that are on the safe road can infer the
state $\theta _{t-2}$ from the flow observed at time $t-1$. For this reason,
from here on we consider recommendation schemes where the CP conditions his
recommendations not only on $\theta _{t-1}$ but also on $\theta _{t-2}$. In
principle, the CP could even condition on further past values of the state $%
\theta $. Though we are not able to rule out formally that conditioning on $%
(\theta_{t-2}, \theta_{t-1})$ is optimal without loss of generality, in what
follows we simplify the analysis of incentive compatible recommendation
schemes by assuming that the CP will condition only on $(\theta
_{t-2},\theta _{t-1})$ and thus the relevant state can be summarized by
equilibrium path beliefs $(\beta _{t-2},\beta _{t-1})$. Based on Proposition %
\ref{prop:so}, we also restrict attention to schemes which do involve
experimentation for all sample paths (meaning that the CP always prefers to
send one agent on the risky road).

\begin{definition}[Infinite horizon recommendation scheme]
\label{def:am} A recommendation scheme is defined as a map $\pi :[0,1]\times
\lbrack 0,1]\rightarrow \{1,2,...,N\}$ which maps the belief of the CP on
the state of the risky road at time $t-2$ and $t-1$ (i.e. $\beta
_{t-2},\beta _{t-1}$) to the number of agents to whom the CP sends a
recommendation to take the risky road, $r_{R}$. If $\theta _{t-1}=L$, we
assume that all agents that were on the risky road at time $t-1$ receive a
recommendation to remain on the risky road at time $t$; the remaining
recommendations (i.e. $\pi (\cdot ,L)-x_{t-1}^{R}$) are sent to a random
subset of the agents on safe. If instead $\theta _{t-1}=H$ then
recommendations are sent to a random subset of all the agents. In both cases
agents that do not receive a recommendation of risky receive a
recommendation of safe, $r_{S}$. Finally, if any agent deviates, the CP
provides full information to all agents from then on.\footnote{%
Intuitively, full information is the worst incentive compatible punishment,
for deviation, that the CP can impose. In fact under full information, the
expected cost per round of each agent is $S_{0}$. No punishment can lead to
higher cost and be incentive compatible because agents can always switch to
play safe and achieve a cost of $S_{0}$.} We denote the set of all
recommendation schemes of this form as $\hat{\Pi}$. 
\end{definition}

{We want to stress that the assumption of restricted history applies only to
the CP, not to agents. Consequently, we are in no way restricting the
optimal behavior of the agents, who can condition their actions on all of
their past information.}

\noindent Any scheme $\pi \in \hat{\Pi}$ can be parametrized as follows 
\begin{equation*}
\pi (\beta _{t-2},\beta _{t-1})=%
\begin{cases}
a & \quad \QTR{up}{if}\ [\beta _{t-2},\beta _{t-1}]=[L,H] \\ 
b & \quad \QTR{up}{if}\ [\beta _{t-2},\beta _{t-1}]=[H,H] \\ 
c & \quad \QTR{up}{if}\ [\beta _{t-2},\beta _{t-1}]=[H,L] \\ 
d & \quad \QTR{up}{if}\ [\beta _{t-2},\beta _{t-1}]=[L,L].%
\end{cases}%
\end{equation*}%
Note that one should also specify the values of $\pi (\beta _{t-2},\beta
_{t-1})$ for values of $\beta_{1},\beta _{2}\in (0,1)$ but in the schemes we
consider this is not relevant in view of the fact that the CP always sends
at least one agent to experiment and thus knows the state of the risky road
in the previous period. For simplicity we denote a generic scheme $\pi $ of
this form as $\pi _{a,b,c,d}$ and the associated social cost as $V_{a,b,c,d}$
(we consider costs starting from $\theta _{0}=H$ since this
choice induces the lowest possible belief and is thus the most difficult
scenario for experimentation).

Example~\ref{example:infty_notic} showed that $\pi^{\QTR{up}{SO}}$ may not
be incentive compatible because if agents on the safe road see the flow $%
N-x_{L}^{\QTR{up}{SO}}$ they can infer that the road switched to $L$ at time 
$t-2$ and may have an incentive to deviate. This intuition motivates us to
focus in particular on a subclass of $\hat{\Pi}$ consisting of schemes
obtained by the following modification of the social optimum policy (see
also Table \ref{fig:scheme_so}). If $\theta _{t-1}=H$ the CP sends one agent
(to experiment) exactly as in the social optimum (i.e. we set $a=b=1$). If
instead $\theta _{t-1}=L$ then the CP sends two possibly different flows $c$
and $d$ depending on whether the road has just switched to $L$ or whether it
was $L$ also in the previous period (in which case the agents on safe can
infer the road changed at time $t-2$ from the flow observed at $t-1$).

\begin{definition}
\label{ab1} Consider a scheme $\pi_{a,b,c,d} \in \hat{\Pi}$ with $a=b=1$ and 
$1 < c \leq d$ and denote this for simplicity as $\pi_{c,d}$.
\end{definition}

Since $c>1$, under any scheme $\pi_{c,d}$, agents can learn the state of the
risky road at $t-2$ by observing the flow at $t-1$. Exploiting this fact, we
show that agents can summarize their history $h_{t-1}$ with a smaller state $%
z_t^i$, as detailed next.

\begin{lemma}
\label{lemma:agent_state} Under $\pi_{c,d}$ and given that
other agents follow their recommendations, an agent can evaluate if a
recommendation is incentive compatible using only the information
\begin{equation*}
z_t^i := [x_{t-1}, \beta^i_{t-1}, r_{t-1}^i]
\end{equation*}
where

\begin{itemize}
\item $x_{t-1} \in X:=\{0, 1, ..., N \}$ is the flow observed on the risky
road at the previous time (even if an agent is on safe he can infer $x_{t-1}$
as $N$ minus the flow on safe, hence this is common information);

\item $\beta^i_{t-1} \in \mathcal{B} :=\{ L, H, U\}$ encodes the information
that agent $i$ has about the state of the road at $t-1$. If the agent was on
the risky road at $t-1$, then he knows the true realization ($L$ or $H$),
while if he was on the safe road we denote the fact that he does not know
the state with the symbol $U$ (Unobserved);\footnote{%
We simply use the symbol $U$ instead of specifying the belief with a number
in $(0,1)$.}

\item $r_{t-1}^i \in \Lambda := \{r_S, r_R \}$ is the recommendation an
agent receives between round $t-1$ and $t$.
\end{itemize}

Specifically, let $h_{t-1}^i$ be the entire history of the agent
up to and including time $t-1$. For any $\alpha^i_t\in\{S,R\}$ it holds
\begin{align*}
\mathbb{E}  &\left[ \tilde{g}(\alpha^i_t, \pi_{c,d}^{-i,t}, \theta_t) +
\sum_{k=t+1}^{\infty} \delta^{k-t} \tilde{g} ( \pi_{c,d}^{i,k} ,
\pi_{c,d}^{-i,k}, \theta_k) \mid h^i_{t-1} \right] \\
&= \mathbb{E}  \left[ \tilde{g} (\alpha^i_t, \pi_{c,d}^{-i,t}, \theta_t) +
\sum_{k=t+1}^{\infty} \delta^{k-t} \tilde{g} ( \pi_{c,d}^{i,k},
\pi_{c,d}^{-i,k}, \theta_k) \mid z^i_{t} \right]
\end{align*}
where $\pi^{i,k}_{c,d}$ is the recommendation sent at time $k>t$ by the CP
if agent $i$ takes action $\alpha^i_t$ at time $t$ and follows the
recommendations from there on, while $\pi^{-i,k}_{c,d}$ denotes the
recommendations sent to all other agents.
\end{lemma}

Intuitively, the flow $x_{t-1}$ is a summary of all that happened up to $%
\theta _{t-2}$ (this is common information) and the combination of $\beta
_{t-1}^{i}$ and $r_{t-1}^{i}$ adds personalized information about an agent's
knowledge of $\theta _{t-1}$ before time $t$. {\ Note that if road
congestion were unobserved ($U$), under $\pi _{c,d}$ the combination of $%
x_{t-1}$ and $r_{t-1}^{i}$ would be enough to provide a unique belief on the
state of the risky road. That is any agents with the same state $%
z_{t}^{j}=z_{t}^{i}$ have the same belief on the state of the risky road.}

Our first main result is to derive sufficient conditions on $c,d$ so that $%
\pi_{c,d}$ is incentive compatible.

\begin{proposition}[Symmetric equilibrium]
\label{prop:eq}\ Suppose that Assumptions \ref{ass:gamma} and \ref%
{ass:mul_muh} hold. Additionally, assume that $c, d$ are such that

\begin{enumerate}
\item $x_L^{\QTR{up}{SO}}\le c\le d\le x_L^{\QTR{up}{eq}}$

\item $g(c,\mu_L)\le g(2,\mu_L)$

\item the pair $(c,d)$ is such that agents that are on safe and receive a
recommendation of $r_{S}$ after observing flow $d$ on risky will follow the
recommendation, that is, 
\begin{equation} \label{eq:def_xll_simple}
\underbrace{u([d,U,r_{S}])}_{\textup{cost of following}}\leq \underbrace{%
p_{d,S}\mu _{L}(d+1)+(1-p_{d,S})2\mu _{H}}_{\substack{ \textup{%
expected stage cost of deviating} \\ \textup{to risky}}}+\underbrace{\frac{%
\delta }{1-\delta }S_{0}}_{\substack{ \textup{continuation cost} \\ \textup{%
of deviating}}}
\end{equation}%
where $p_{d,S}=\mathbb{P}(\theta _{t-1}=L\mid z_{t}^{i}=[d,U,r_{S}])$.{\ The
constraint \eqref{eq:def_xll_simple} is written implicitly for readability
and an explicit formula is provided in \eqref{eq:def_xll} in the Appendix.}
\end{enumerate}

Then, the recommendation scheme $\pi_{c,d}$ induces the symmetric
equilibrium 
\begin{align}  \label{type1:H}
\xi^i_{ \pi_{c,d}}(z^i_t) = \xi^i_{\pi_{c,d}}([x_{t-1}, \beta^i_{t-1},
r^i_{t-1}]) = 
\begin{cases}
R & \text{if}\ r^i_{t-1} = r_R, \\ 
S & \text{otherwise},%
\end{cases}%
\end{align}
and is thus incentive compatible.
\end{proposition}

The intuition behind the conditions derived in the Proposition \ref{prop:eq}
are given next:

\begin{enumerate}
\item after the road switches to $L$ for the first time the CP sends at
least the social optimum number of agents and he possibly increases the flow
after that, but no more than the myopic equilibrium flow;

\item the flow sent by the CP on the risky road after the road switches to $L
$ for the first time leads to a no worse stage cost than sending just two
agents;

\item $d$ is large enough so that agents that are on safe and infer $%
\theta_{t-2}=L$ follow the recommendation to remain on safe (thus addressing
the issue identified in Example~\ref{example:infty_notic}).
\end{enumerate}

The proof of Proposition \ref{prop:eq} is presented in the Appendix. Here we
provide some intuition. The first fundamental observation that we make is
that, since the experimenter after $\theta _{t-1}=H$ is chosen at random
among all the agents, each agent has the same continuation cost (which we
term $\bar{v}$) after he observes the risky road switching from $L$ to $H$.
Because of this we can divide the infinite horizon into periods (by defining
the beginning of a new period as the time immediately after the risky road
switches from $L$ to $H$) and study incentive compatibility only until the
end of the current period. The division in periods and the number of agents
taking the risky road under the schemes $\pi^{\QTR{up}{SO}}$ and $\pi _{c,d}$
for each period are illustrated in Table \ref{fig:scheme_so}.

\begin{table}[]
\begin{center}
$\hdots L H | \underbrace{H \hdots H L \hdots L H}_{period} | \underbrace{H %
\hdots H L \hdots L H}_{period} | \underbrace{H \hdots H L \hdots L H}%
_{period}| \underbrace{H \hdots H L \hdots L H}_{period} | H \hdots$\\[0.3cm]
\begin{tabular}{|c|c|cccc|c|cccc|}
\hline
State of the risky road & H & H & $\hdots$ & H & L & L & L & $\hdots$ & L & H
\\[0.1cm] \hline
$\pi^{\QTR{up}{SO}}$ (Social optimum) & - & 1 & $\hdots$ & 1 & 1 & $x_{L}^{%
\QTR{up}{SO}}$ & $x_{L}^{\QTR{up}{SO}}$ & $\hdots$ & $x_{L}^{\QTR{up}{SO}}$
& $x_{L}^{\QTR{up}{SO}}$ \\[0.1cm] 
$\pi_{c,d}$ (Proposition \ref{prop:eq}) & - & 1 & $\hdots$ & 1 & 1 & $c$ & $%
d $ & $\hdots$ & $d$ & $d$ \\[0.1cm] \hline
\end{tabular}%
\end{center}
\caption{{Comparison between the flows on the risky road under schemes $%
\protect\pi^{\textup{SO}}$ and $\protect\pi_{c,d}$ for one period.}}
\label{fig:scheme_so}
\end{table}

To prove incentive compatibility of $\pi _{c,d}$ we then need to show that
no agent can improve his cost by a unilateral deviation. To this end, we
divide the agents into four types:

\begin{enumerate}
\item \ul{Agents that took the risky road at time $t-1$ and saw $L$}: under $\pi _{c,d}$ these agents receive a recommendation of risky. Since the
probability that the road changes from $L$ to $H$ in one step is $\gamma
_{L}\leq \frac{1}{2}$, one should expect that following such a
recommendation is incentive compatible. In particular, we show that the
stage cost obtained by following the recommendation is less than $S_{0}$ (as
proven in Lemma \ref{lemma:mlxl} in the Appendix), hence any deviation will
increase both current cost and also continuation cost (because it leads to
lower information than using the risky road and is thus not profitable).

\item \underline{that took the risky road at time $t-1$ and saw $H$}: there
are two cases, either the agent receives a recommendation to take the safe
road (which is intuitively incentive compatible since the probability that
the road changes from $H$ to $L$ in one step is $\gamma _{H}\leq \frac{1}{2}$) or the agent receives a recommendation to take the risky road. The only
case when the latter happens is if the agent is selected to be the next
experimenter. In this case we show that, even though experimentation is
costly in terms of current payoffs, the continuation cost is lower from
experimenting than from deviating (recall that after any deviation the CP
provides full information in all future periods). This makes being the
experimenter incentive compatible.

\item \ul{Agents that took the safe road at time $t-1$ and received a
recommendation to remain on the safe road}: as noted in Example~\ref%
{example:infty_notic} from observing $x_{t-1}=c>1$ or $x_{t-1}=d>1$ these
agents can infer $\theta _{t-2}=L$. Condition \eqref{eq:def_xll_simple}
guarantees it is incentive compatible for an agent that observed $x_{t-1}=d$
to follow a recommendation of using the safe road. We show that this
condition implies incentive compatibility also for the case $x_{t-1}=c$. The
only remaining possibility is when $x_{t-1}=1$, in this case the agent can
infer $\theta _{t-2}=H$ and incentive compatibility is immediate.

\item 
\ul{Agents that took the safe road at time $t-1$ and received a
recommendation to take the risky road}: incentive compatibility in this case
follows with the same argument as in cases 1 and 2. Indeed, the agent has
either been chosen to benefit from using the risky road when the state is $L$
(which is incentive compatible by the discussion for case 1) or he has been
chosen as an experimenter (which is incentive compatible by the discussion
for case 2).
\end{enumerate}

\subsection{Partial information: Optimality}

Motivated by Example \ref{example:infty_notic}, we consider a specific
scheme among those that are incentive compatible according to Proposition %
\ref{prop:eq}. Specifically, for the period immediately after the road
switched from $H$ to $L$ we assume that the CP sends the flow $c=x_L^{%
\QTR{up}{SO}}$ exactly as in the social optimum (intuitively this is
possible, because agents on safe are unaware that the road condition
changed). For all subsequent periods the CP sends the minimum number of
agents to maintain incentive compatibility (i.e. to satisfy %
\eqref{eq:def_xll_simple} for $c=x_L^{\QTR{up}{SO}}$). We denote this flow
by $x_{LL}$.

\begin{definition}
We define the scheme $\pi^*\in \hat{\Pi}$ as follows 
\begin{align}
\pi^*(\beta_{t-2}, \beta_{t-1}) = 
\begin{cases}
1 & \text{if}\ \beta_{t-1} = H, \\ 
x_L^{\QTR{up}{SO}} & \text{if}\ \beta_{t-2} = H, \beta_{t-1} = L, \\ 
x_{LL} & \text{if}\ \beta_{t-2} = L, \beta_{t-1} = L,%
\end{cases}%
\end{align}
where $x_{LL} := \max \{x_L^{\QTR{up}{SO}}, \bar{x}_{LL} \}$ with $\bar{x}%
_{LL}$ being the smallest integer such that $d=\bar{x}_{LL}$ satisfies %
\eqref{eq:def_xll_simple} for $c=x_L^{\QTR{up}{SO}}$. In other words, $%
\pi^*=\pi_{(x_L^{\QTR{up}{SO}},x_{LL})}$.
\end{definition}

\begin{corollary}
\label{corollary:star_ic} Suppose that Assumptions \ref{ass:gamma} and \ref%
{ass:mul_muh} hold. The scheme $\pi^*\in \hat{\Pi}$ is incentive compatible
and achieves strictly lower social cost than full information.
\end{corollary}

This corollary follows immediately from Proposition \ref{prop:eq} upon
noting that the pair $c=x_{L}^{\QTR{up}{SO}}$ and $d=x_{LL}$ satisfy the
assumptions of that proposition (we prove in Lemma \ref{lemma:xll} in the
Appendix that $\bar{x}_{LL}\leq x_{L}^{\QTR{up}{eq}}$). The fact that the
social cost is strictly less than full information is proven in point 1 of
Lemma \ref{lemma:punish} (in the Appendix).

We next derive sufficient conditions for the scheme $\pi^*$ to be not only
incentive compatible, but also optimal. We consider two different regimes
depending on the discount factor $\delta$ used by the agents to weight
future travel times.

\subsubsection{Large $\protect\delta$}

We show that as $\delta \rightarrow 1$, $x_{LL}\rightarrow x_{L}^{\QTR{up}{SO%
}}$. In other words, the cost under $\pi ^{\ast }$ converges to the cost of
the social optimum as $\delta \rightarrow 1$. Define the social cost
starting from belief $\beta =H$ under the social optimum and $\pi ^{\ast }$
as $V_{H}^{\QTR{up}{SO}}$ and $V_{H}^{\pi ^{\ast }}$, respectively.

\begin{proposition}
\label{prop:ratio} Suppose that Assumptions \ref{ass:gamma} and \ref%
{ass:mul_muh} hold and assume $\gamma_{H},\gamma _{L}>0$. Then the cost
under $\pi ^{\ast }$ approaches the social optimum as $\delta \rightarrow 1$%
. Formally, 
\begin{equation*}
\lim_{\delta \rightarrow 1}\frac{V_{H}^{\pi ^{\ast }}}{V_{H}^{\QTR{up}{SO}}}%
=1.
\end{equation*}
\end{proposition}

This full optimality result obtains because for large $\delta $ the policy $%
\pi _{(c,d)}=\pi _{(x_{L}^{\QTR{up}{SO}},x_{L}^{\QTR{up}{SO}})}$ satisfies %
\eqref{eq:def_xll_simple} hence $\pi ^{\ast }$ coincides with $\pi ^{%
\QTR{up}{SO}}$. This result is to be expected. In fact given any time $t$
let $t^{H}$ be the first time the road switches to $\theta =H$ after $t$
(this event happens in finite time since $\gamma _{L}>0$). Then under any
policy $\pi _{(c,d)}$, the cost of any agent is 
\begin{equation*}
\sum_{\tau =t+1}^{t_{H}}\delta ^{\tau -t}\text{cost}_{\tau }+\delta ^{\tau
-t_{H}}\frac{V_{H}^{\pi _{(c,d)}}}{N}.
\end{equation*}%
For $\delta \rightarrow 1$ the first term is negligible; hence any scheme
for which $V_{H}^{\pi }<S_{0}N/(1-\delta )$ (i.e. the cost under full
information) is incentive compatible. Clearly the social optimum meets this
condition and hence it must be incentive compatible.

\subsubsection{Small $\protect\delta$}

Before stating our main result we show that for $\delta$ small, for any
scheme $\pi\in\hat\Pi$ to be incentive compatible it must be $a=b=1$ (thus
justifying our interest in the class of schemes given in Definition \ref{ab1}%
).

\begin{lemma}
\label{lem:small} Suppose that $\delta\le \frac12$ and that Assumptions \ref%
{ass:gamma} and \ref{ass:mul_muh} hold. Then if a scheme $\pi\in\hat\Pi$ is
incentive compatible it must be $a=b=1$.
\end{lemma}

{The intuition for this result is straightforward. Since, after seeing $H$,
the CP selects the next experimenter randomly, in any scheme $\pi $ there is
a positive probability that the selected experimenter knows that the risky
road was $H$ in the previous period. If either $a$ or $b$ are greater than
one, then the expected cost of the experimenter would be greater than $2\mu
_{H}$ (which is the expected stage cost). On the other hand, if the
experimenter deviates, the CP provides full information and can guarantee an
expected cost of {$\frac{S_{0}}{1-\delta }$ from then on. }Under Assumption %
\ref{ass:mul_muh}, and if $\delta \leq \frac{1}{2}$, 
\begin{equation*}
\frac{S_{0}}{1-\delta }\leq 2S_{0}\leq 2\mu _{H}.
\end{equation*}%
Hence having more than one experimenter cannot be incentive compatible.%
\footnote{%
We note that instead if $\delta >\frac{1}{2}$, a scheme with $a$ or $b$
greater than one, might be incentive compatible. Although, sending more than
one agent to experiment always gives a higher stage cost for the CP, it is
unclear under higher $\delta $ whether it may benefit the CP to send a
higher flow after $H$ to drive down $d$ either through raising the cost of
deviation in this setting or through obfuscation of information by making
the flow the same after $H$ and $L$. Overall, when $\delta >\frac{1}{2}$ we
are unable to rule out that a scheme sending more than one agent to
experiment could be incentive compatible and give a lower overall cost.}
Having fixed $a,b$ we now turn to the optimal choice of $c,d$. }

\begin{proposition}
\label{prop:delta_small} Suppose that Assumptions \ref{ass:gamma} and \ref%
{ass:mul_muh} hold and $N \geq 5$. Then

\begin{enumerate}
\item for $\delta$ sufficiently small, $\pi^*$ achieves the minimum social
cost among all the incentive compatible schemes belonging to $\hat{\Pi}$;

\item for $\delta\le \frac12$, the scheme that minimizes the social cost
among all the incentive compatible schemes belonging to $\hat{\Pi}$ is
either $\pi^*$ or $\tilde{\pi}^*:=\pi_{x_{L}^{\QTR{up}{SO}}+1,x_{LL}-1}$.
\end{enumerate}
\end{proposition}

{To understand the previous result recall that the social optimum choice
would be $c=d=x_{L}^{\QTR{up}{SO}}$. Unfortunately, in most cases this
choice is not incentive compatible because of constraint %
\eqref{eq:def_xll_simple} (guaranteeing that agents that are on safe follow
a recommendation of safe). Our first step in the proof of Proposition \ref%
{prop:delta_small} is to show that the incentive compatibility constraint %
\eqref{eq:def_xll_simple} can be rewritten as $f(c)\leq g(d)$ where $f(c)$
is convex in $c$ and is minimized at a value between $x_{L}^{\QTR{up}{SO}}$
and $x_{L}^{\QTR{up}{SO}}+1$. By the integer nature of our problem, this
immediately implies that at optimality $c$ should take one of this two
values (for having a larger value of $c$ would make the constraint %
\eqref{eq:def_xll_simple} harder to satisfy (leading to $d\geq x_{LL}$) and
would thus lead to a scheme with higher social cost than $\pi ^{\ast }$;
recall that $c=x_{L}^{\QTR{up}{SO}}$ would be the optimal choice to minimize
the social cost). }

If the minimizer is $c=x_{L}^{\QTR{up}{SO}}$, then by definition it must be $%
d\geq x_{LL}$. If, instead, $c=x_{L}^{\QTR{up}{SO}}+1$ is the minimizer then
there exist parameters under which $d=x_{LL}-1$ can be incentive compatible
(but we show that no smaller values of $d$ could be).\footnote{%
In fact, in some cases increasing the value of $c$ leads to a smaller
continuation cost for agents that follow the recommendation of safe (since
they have higher probability to be sent to the risky road when the road
changes to $L$ in the next period). When that happens $\tilde{\pi}^{*}$ is
incentive compatible by \hbox{Proposition \ref{prop:eq}.}} The pair $%
c=x_{L}^{\QTR{up}{SO}}+1$ and $d=x_{LL}-1$, defining $\tilde{\pi}^{\ast }$,
may give a lower social cost than $\pi ^{\ast }$ for certain parameter
values (because there are more future rounds with flow $d$ in expectation
than rounds with $c$, hence increasing $c$ slightly to decrease $d$ may be
beneficial).\footnote{%
While $\tilde{\pi}^{*}$ has a higher cost that $\pi ^{*}$ immediately after
the road switches to $L$ (since $g(x_{L}^{\QTR{up}{SO}}+1,\mu _{L})>g(x_{L}^{%
\QTR{up}{SO}},\mu _{L})$), it has lower stage cost for all the subsequent
times (since $g(x_{LL}-1,\mu _{L})<g(x_{LL},\mu _{L})$). For sufficiently
large values of $\delta $ this may reduce the overall cost.} The second
statement of Proposition \ref{prop:delta_small} follows immediately from
these observations. The first statement follows from the observation that,
for $\delta $ small enough, the scheme $\pi ^{\ast }$ must have smaller
social cost than $\tilde{\pi}^{\ast }$ (since it leads to smaller cost for
the stage immediately after the state of the road switches to $L$, and for $%
\delta $ small enough, this dominates the potential future gain of using $%
d=x_{LL}-1$ instead of $d=x_{LL}$).

\section{Conclusion}

New GPS technologies and traffic recommendation systems critically depend on
real-time information about road conditions and delays on a large number of
routes. This information mainly comes from the experiences of drivers.
Consequently, enough drivers have to be induced to experiment with different
roads (even if this involves worse expected travel times for them). This
situation creates a classic experimentation-exploitation trade-off, but
critically\ one in which the party interested in acquiring new information
cannot directly choose to experiment but has to convince selfish, autonomous
agents to do it. This is the problem we investigate in the current paper.

There is by now a large literature on experimentation in economics and
operations research. The main focus is on the optimal amount of
experimentation by trying new or less well-known options in order to acquire
information at the expense of foregoing current high payoffs. The game
theoretic experimentation literature, investigating situations where there
are multiple agents who can generate information for themselves and others,
studies issues of collective learning, free-riding and underexperimentation.
Missing from the previous literature is the main focus of our paper: a
setting in which exploitation of relevant information creates payoff
dependence (for example, via congestion in the context of our routing model)
and the central entity or planner has the incentives for experimentation,
but has to confront the incentive compatibility of the agents, especially in
view of the aforementioned payoff dependence.

We develop a simple model to study these issues, and characterize optimal
recommendation systems first in a two-stage setting and then in an
infinite-horizon environment. Key aspects of our model are congestion
externalities on roads (introducing payoff dependence); a finite number of
agents (so that agents take into account their impact on information as well
as congestion); forward-looking behavior by agents (so that they can be
incentivized by future rewards); and a central planner who can observe
results from experimentation and can make recommendations but has to respect
incentive compatibility (introducing the feature that this is not a direct
model of experimentation). We simplify our analysis by assuming that there
are only two roads and one of them is\ \textquotedblleft
safe\textquotedblright , meaning that the travel time is known,
non-stochastic, and does not depend on the state of nature. This contrasts
with the other, \textquotedblleft risky\textquotedblright\ road, where
travel times depend on the state of nature (on which the central planner is
acquiring information).

We first show that full information, whereby the central planner shares all
the information he acquires with all agents, is generally not optimal. The
reason is instructive about the forces in our model: full information will
make all agents exploit information about favorable conditions on the risky
road, and this will in turn cause congestion on this risky road, reducing
the rewards the experimenter would need to reap in order to encourage his
experimentation. As a result, full information may lead to insufficient or
no experimentation, which is socially costly.

We then proceed to characterizing optimal incentive compatible
recommendation schemes. These typically do not induce full information, but
still share some of the information obtained from the experimentation of few
experimenters (in our model only one experimenter is sufficient because
there is no uncertainty about the state conditional on experimentation).

In the case of infinite-horizon, the underlying state of the risky road
changes according to a Markov chain. An additional issue in this case is
that the incentive compatibility of non-informed agents has to be ensured as
well, since they may decide to disregard the recommendation of the central
planner and choose the risky road when they think travel times are lower
there. This makes the characterization of the optimal recommendation scheme
more challenging. We propose a relatively simple incentive compatible
dynamic scheme and then establish its optimality when the discount factor is
small (in particular less than 1/2) and large enough (limiting to 1).

Our paper highlights the importance of understanding how modern routing
technologies (and perhaps more generally) need to induce sufficient
experimentation and how they can balance the benefits from exploiting new
information and ensuring incentive compatibility of experimentation as well
as incentive compatibility of all non-experimenters. Investigating how these
issues can be navigated in more general settings (for example, with a more
realistic road network and richer dynamic and stochastic elements or in
models of payoff dependence resulting from other considerations) is an
important area for future work.

\newpage \onecolumn\appendix

\begin{center}
\textbf{APPENDIX}
\end{center}

\section{Proofs of Section 3: Two stage example}

\begin{proof}
\textbf{of Lemma \ref{lemma:pure_one_exp}} Suppose that $k \geq 2$ agents
experiment. Then the cost of any of these agents is 
\begin{equation*}
[\textup{expected cost of risky}] = k \mu_{\beta} + \eta_R \ge k \mu_{\beta}
\geq 2 \mu_{\beta} > 2(S_0 + S_1 N), 
\end{equation*}
where we denoted by $\eta_R$ the expected cost in the second round, which is
for sure non-negative and we used Assumption \ref{assumption:2stage}. If
instead the agent switches to safe he will have an expected cost of 
\begin{equation*}
[\textup{expected cost of safe}] = S_0 + S_1(N-k+1) + \eta_S \le S_0 +
S_1(N-1) + S_0 + S_1 N, 
\end{equation*}
where we denoted by $\eta_S$ the expected cost in the second round, which is
at most $S_0 + S_1 N$ as the agent can always play safe in the second round
and in the worst case every other agent is also playing safe. Since $S_0 +
S_1(N-1) + S_0 + S_1 N< 2(S_0 + S_1 N)$, it follows that the cost of
experimenting is greater than the overall cost of taking safe. Thus, it is
never a pure strategy equilibrium for more than one person to experiment --
under any information scheme.
\end{proof}
\begin{proof}
\textbf{of Theorem \ref{thrm:private_vs_full}} Recall we are considering
pure strategy equilibria. Using Lemma~\ref{lemma:pure_one_exp} we can
characterize when it is an agent's best response to experiment under the two
different information schemes.

\textbf{Full information}: If no agent experiments the cost for each agent
is $2(S_0 + S_1 N)$. If any agent experiments all agents learn $\theta$
before round two and there are two possibilities.

\begin{itemize}
\item[1)] If $\theta = H$, everyone takes the safe road in the second round
as $S_0 + S_1 N < H$.

\item[2)] If $\theta = L$, the agents play a pure strategy Nash equilibrium
and split the flow across the two roads, i.e. $x_L^{\textup{eq}}$ take risky. The
expected cost of equilibrium in the second round when $\theta = L$ is
therefore  $g(x_L^{\textup{eq}},L)/N = (x_L^{\textup{eq}}/N) (L x_L^{%
\textup{eq}}) + (N-x_L^{\textup{eq}})/N(S_0 + S_1 (N-x_L^{\textup{eq}}))$.
\end{itemize}

All agents playing safe is an equilibrium if and only if the cost of
switching to experimenting is worse than $2(S_0 + S_1 N)$. The expected cost
of one agent switching is 
\begin{equation*}
\mu_{\beta} + \beta g(x_L^{\textup{eq}},L)/N + (1-\beta) (S_0 + S_1 N) 
\end{equation*}
where the first term is the cost of experimenting in the first round, the
second term is the cost in the second round if $\theta = L$ weighted by $%
\mathbb{P} (\theta = L)$, and similarly the last term is the cost if $\theta
= H$ weighted by $\mathbb{P} (\theta = H)$.

Overall, if 
\begin{equation}  \label{ICA}
2(S_0 + S_1 N)<\mu_{\beta} + \beta g(x_L^{\textup{eq}},L)/N + (1-\beta) (S_0
+ S_1 N),
\end{equation}
no one experiments. Otherwise, one agent experimenting in the first round is
the unique pure strategy Nash equilibrium (recall by Lemma \ref%
{lemma:pure_one_exp} that it is never incentive compatible for more than one
agent to experiment). The total cost of the first round under this
equilibrium is $g(1, \mu_{\beta})$ and if $\theta = H$ the total cost of the
second round is $g(0)$, while if $\theta = L$ the total cost of the second
round is $g(x_L^{\textup{eq}}, L)$. The thresholds $\beta_f$ can be obtained by
imposing equality in \eqref{ICA} and solving for $\beta$.

\textbf{Private information}: No experimentation is an equilibrium if and
only if it is not individually optimal for an agent to play risky. Similarly
to the previous case, this occurs when the expected cost of switching to
playing risky is worse than all agents playing safe, that is, when 
\begin{equation}  \label{ICB}
2(S_0 + S_1 N) < \mu_{\beta} + \beta L + (1-\beta)(S_0 + S_1 N),
\end{equation}
where we used the fact that, under private information, it is a best
response for all the agents that were on safe at time $1$ to remain on safe
at time $2$, while for the experimenter it is a best response to take risky
at time $2$ if he observed $L$ at time $1$ and safe otherwise.

If the above does not hold, then there is an incentive to deviate from all
playing safe and there exists an asymmetric pure strategy Nash equilibrium,
where one agent experiments in the first round. By Lemma \ref%
{lemma:pure_one_exp}, this is the unique pure strategy Nash equilibrium. The
expected cost of the first round is $g(1, \mu_{\beta})$ and the expected cost
of the second round is $\beta g(1, L) + (1-\beta) g(0).$ The thresholds $%
\beta_p$ can be obtained by imposing equality in \eqref{ICB} and solving for 
$\beta$.

Finally, note that $\beta_p\le \beta_f$ if and only if $\frac{g(x_L^{%
\textup{eq}}, L)}{N} \geq L$. Note that by Assumption \ref{assumption:2stage} ${L < S_0 + S_1}$, thus $x_L^{\textup{eq}} \geq 1$. If $x_L^{\textup{eq}} = N$ then $g(x_L^{\textup{eq}}, L) = LN^2$ and the inequality holds. Instead, if $1 \leq x_L^{\textup{eq}} \le N-1$ then 
\begin{align*}
g(x_L^{\textup{eq}},L) &= L (x_L^{\textup{eq}})^2 + (S_0 + S_1 (N-x_L^{\textup{eq}}))(N-x_L^{\textup{eq}})\\
&= L (x_L^{\textup{eq}})^2 + S_0  (N-x_L^{\textup{eq}}) + S_1 (N-x_L^{\textup{eq}})(N-x_L^{\textup{eq}})\\
&\ge L x_L^{\textup{eq}} + (S_0 +S_1) (N-x_L^{\textup{eq}}) \\
&> L x_L^{eq} + L (N-x_L^{\textup{eq}})\\
&\geq LN
\end{align*}
where the first inequality follows $x_L^{\textup{eq}},N-x_L^{\textup{eq}}\ge 1$, while the second follows from $L < S_0 + S_1$.
\end{proof}

\begin{proof}
\textbf{of Corollary \ref{corollary:private_better}} In this interval of
beliefs:

\begin{itemize}
\item under full information there is no experimentation and the cost is $%
2g(0)$.

\item under private information there is experimentation. The experimenter
has a total cost for the two rounds that is less than $2(S_0+S_1N)$,
otherwise he would switch to safe. All the other agents have cost $%
2(S_0+S_1(N-1))< 2(S_0+S_1N).$ Therefore the total cost is strictly less
than $2N(S_0+S_1N)=2g(0)$.
\end{itemize}
\end{proof}

\begin{proof}
\textbf{of Theorem \ref{lemma:so_2stage}} The social optimum is the minimum
total cost for the two rounds. If all agents use the safe road in both
rounds then the total cost is $2g(0)$. If the CP sends \textit{at least one
agent} on the risky road in the first round the CP learns $\theta$ and can
make a decision on how many agents to send on the risky road in the second
round based on the value of $\theta$. Denote these flows by $x_2^{R\mid L}$
and $x_2^{R\mid H}$. Thus, if the CP  experiments with $x_1^R>0$
agents in the first round, he is facing the following optimization problem 
\begin{align*}
\min_{x_1^R, x_2^{R\mid L}, x_2^{R\mid H}} g(x_1^R, \mu_{\beta}) &+ \beta
g(x_2^{R\mid L}, L) + (1-\beta) g(x_2^{R\mid H}, H) \\
\text{s.t.}\qquad 1 &\leq x_1^R \leq N \\
0 &\leq x_2^{R\mid L} \leq N \\
0 &\leq x_2^{R\mid H} \leq N \\
&x_1^R, x_2^{R\mid L}, x_2^{R\mid H} \in \mathbb{Z} _{\geq 0}.
\end{align*}
This minimization can be separated into three optimization problems, one for
each of the flows $x_1^R, x_2^{R\mid L}, x_2^{R\mid H}$.

\begin{itemize}
\item For the first round $x_1^R$ can be obtained by solving 
\begin{align*}
&\min_{x_1^R} \quad \mu_{\beta} (x_1^R)^2 + (S_0 + S_1(N-x_1^R))(N-x_1^R) \\
&\text{s.t.}\qquad 1 \leq x_1^R \leq N \\
&\qquad \qquad x_1^R \in \mathbb{Z} _{\geq 0}.
\end{align*}
Since $S_0 + S_1 N < \mu_{\beta}$, the cost is strictly increasing for $%
x_1^R\ge 1$ and thus the optimal solution is $x_1^R = 1$.

\item For the second round, the two values $x_2^{R\mid L}$ and $x_2^{R\mid H}
$ can be found separately and are just the values that minimize $g(x, L)$
and $g(x, H)$ respectively, which are $x_L^{\textup{SO}}$ and $x_H^{\textup{SO}}$ by
definition.
\end{itemize}

The minimum of all agents playing safe and the objective of the above
minimization problem gives the social optimum cost and the threshold $\beta_{%
\text{SO}}$ is the belief under which the CP is indifferent between
experimentation and all agents playing safe. Note that $\beta_{SO}<\beta_p$
because if experimentation is an equilibrium under private information it
implies it is optimal for the CP to experiment. Specifically if
experimentation is an equilibrium under private information then $\beta$ is
such that 
\begin{equation*}
2 (S_0 + S_1 N) \geq \mu_{\beta} + \beta L + (1-\beta)(S_0 + S_1 N). 
\end{equation*}
This implies 
\begin{align*}
2 (S_0 + S_1 N) &+ 2 (S_0 + S_1 N)(N-1) \\
&\geq \mu_{\beta} + (S_0 + S_1
N)(N-1) + \beta (L + (S_0 + S_1 N)(N-1)) \\
&\qquad \qquad \qquad \qquad + (1-\beta)((S_0 + S_1 N)N) \\
&\geq \mu_{\beta} + (S_0 + S_1(N-1))(N-1) + \beta g(x_L^{\textup{SO}}, L) +
(1-\beta) g(x_H^{\textup{SO}}, H) \\
\iff 2 (S_0 + S_1 N)N &\geq \mu_{\beta} + (S_0 + S_1 (N-1))(N-1) + \beta
g(x_L^{\textup{SO}}, L) + (1-\beta) g(x_H^{\textup{SO}}, H)
\end{align*}
which is the condition for experimentation to be optimal for the CP.
\end{proof}

\begin{proof}
\textbf{of Theorem \ref{thrm:2stageopt}}

Recall by Lemma \ref{lemma:pure_one_exp} that in any equilibrium there is at
most one experimenter in the first round. Hence we can distinguish two cases:

\begin{enumerate}
\item \textbf{No experimentation:} if no agent experiments the social cost
is $2g(0)$;

\item \textbf{One experimenter:} For a recommendation scheme $\pi$ to be
incentive compatible, it must be that $\pi(L)$, $\pi(H)$ are such that an
agent follows the recommendation he is given. We study incentive
compatibility starting from the \textbf{second round}. In this case there
are three type of agents:
\begin{itemize}
\item \textbf{Type 1 (experimenter in round 2)}: The experimenter knows the
value of $\theta$ since he observed it in the first round. We next show that
without loss of optimality we can restrict our attention to recommendation
schemes where it is a best response for the experimenter to take risky in
the second round if $\theta=L$ and safe if $\theta=H$.

\begin{itemize}
\item $\theta = H$: The experimenter's cost on safe is $S_0 + S_1(N-\pi(H))$%
, the cost on risky is $H(\pi(H)+1)$. The conclusion follows since, 
\begin{align*}
S_0 + S_1(N-\pi(H)) \le S_0+S_1N < H \le H(\pi(H)+1),
\end{align*}
where we used the fact that $H> S_0 + S_1 N$ by assumption.

\item $\theta = L$: Let $x_2^{R\mid L}$ be the equilibrium flow on the risky
road in the second round (this a priori may or may not include the
experimenter). For incentive compatibility it must be $x_2^{R\mid L}\le
x^{\textup{eq}}_L$. In fact if that was not the case, consider an agent that was on
safe in the first round and receives a recommendation of risky. This agent
doesn't know $\theta$, but he knows that in both cases ($\theta=L$ or $%
\theta=H$) switching to safe would give a better cost. Hence a scheme that
leads to $x_2^{R\mid L}> x^{\textup{eq}}_L$ is not incentive compatible. We then
distinguish two cases for the experimenter:

\begin{itemize}
\item if the experimenter belongs to the flow $x_2^{R\mid L}$ then deviating
to safe is not convenient because $x_2^{R\mid L}\le x^{\textup{eq}}_L$;

\item if the experimenter is on safe, then either $x_2^{R\mid L}< x^{\textup{eq}}_L$
in which cases it is convenient for the experimenter to deviate to risky or $%
x_2^{R\mid L}=x^{\textup{eq}}_L$.
\end{itemize}

Overall the only case when it might be convenient for the experimenter to
take safe after observing $\theta=L$ is for recommendation schemes such that 
$x_2^{R\mid L}=x^{\textup{eq}}_L$. Note that the social cost of such a scheme is the
same as full information. We are going to show at the end of this proof that
the optimal solution of \eqref{eq:opt2} is weakly less than full information.
\end{itemize}

Overall, the previous discussion shows that we can assume $x_2^{R\mid
L}=\pi(L)+1$ and $x_2^{R\mid H}=\pi(H)$ without loss of optimality. For
simplicity we denote these flows by $x_L$ and $x_H$ in the rest of this
proof.

\item \textbf{Type 2 (recommended safe)}: An agent of this type took safe in
the first round and received a recommendation to take safe, signal $r_S$,
before the second round. His expected cost of following the recommendation
is 
\begin{align*}
S_0 + S_1 (\mathbb{P}(\theta = L \mid r_S)(N-x_L) + \mathbb{P}(\theta = H
\mid r_S)(N-x_H)),
\end{align*}
as the flow he will experience depends on how many agents are being sent to
risky. Deviating gives an expected cost of 
\begin{align*}
\mathbb{P}(\theta = L \mid r_S) L (x_L + 1) + \mathbb{P}(\theta = H \mid
r_S) H (x_H + 1).
\end{align*}
By Bayes rule 
\begin{align*}
\mathbb{P}(\theta = L \mid r_S) &= \frac{\beta \mathbb{P}(r_S \mid \theta=L)%
}{\beta \mathbb{P}(r_S \mid \theta = L) + (1-\beta)\mathbb{P}(r_S\mid \theta
= H)} \\
&= \frac{\beta \frac{N-x_L}{N-1}}{\beta \frac{N-x_L}{N-1}+(1-\beta)\frac{%
N-x_H-1}{N-1}} \\
&= \frac{\beta (N-x_L)}{\beta(N-x_L)+(1-\beta)(N-x_H-1)} \\
\mathbb{P}(\theta = H \mid r_S) &= 1 - \mathbb{P}(\theta=L \mid r_S).
\end{align*}
Thus, the full constraint is 
\begin{align}  \label{agent2:S}
&S_0 + S_1 \left( \frac{\beta (N-x_L)^2}{\beta(N-x_L)+(1-\beta)(N-x_H-1)} + 
\frac{(1-\beta) (N-x_H)(N-x_H-1)}{\beta(N-x_L)+(1-\beta)(N-x_H-1)} \right) \\
&\leq \frac{\beta (N-x_L) L (x_L + 1)}{\beta(N-x_L)+(1-\beta)(N-x_H-1)} + 
\frac{(1-\beta) (N-x_H-1)H (x_H + 1)}{\beta(N-x_L)+(1-\beta)(N-x_H-1)} . 
\notag
\end{align}

\item \textbf{Type 3 (recommended risky)} An agent of this type took safe in
the first round and received a recommendation to take risky, signal $r_R$,
before the second round. His expected cost of following the recommendation
is 
\begin{align*}
\mathbb{P}(\theta = L \mid r_R) L x_L + \mathbb{P}(\theta = H \mid r_R) H x_H
\end{align*}
and deviating gives an expected cost of 
\begin{align*}
S_0 + S_1 (\mathbb{P}(\theta = L \mid r_R)(N-x_L+1) + \mathbb{P}(\theta = H
\mid r_R)(N-x_H+1)).
\end{align*}
By Bayes rule 
\begin{align*}
\mathbb{P}(\theta = L \mid r_R) &= \frac{\beta \mathbb{P}(r_R \mid \theta=L)%
}{\beta \mathbb{P}(r_R \mid \theta = L) + (1-\beta)\mathbb{P}(r_R\mid \theta
= H)} \\
&= \frac{\beta \frac{x_L-1}{N-1}}{\beta \frac{x_L-1}{N-1}+(1-\beta)\frac{x_H%
}{N-1}} \\
&= \frac{\beta (x_L-1)}{\beta(x_L-1)+(1-\beta) x_H} \\
\mathbb{P}(\theta = H \mid r_R) &= 1 - \mathbb{P}(\theta=L \mid r_R).
\end{align*}
Thus, the full constraint is 
\begin{align}  \label{agent3:R}
&\frac{\beta (x_L-1)}{\beta(x_L-1)+(1-\beta)x_H} L x_L + \frac{(1-\beta)x_H}{%
\beta(x_L-1)+(1-\beta)x_H} H x_H \\
&\leq S_0 + S_1 \left( \frac{\beta (x_L-1)}{\beta(x_L-1)+(1-\beta) x_H }%
(N-x_L+1) \right. \\
&\qquad \qquad \qquad \qquad \left. + \frac{(1-\beta)x_H}{\beta(x_L-1)+(1-\beta) x_H}(N-x_H+1) \right)
\notag
\end{align}
\end{itemize}

The constraints above are for the {second round}, we next consider the 
\textbf{first round}. We already know by Lemma \ref{lemma:pure_one_exp} that
it is not convenient for any agent on safe to join the experimenter. Hence
we only need to ensure that it is incentive compatible for the experimenter
to experiment in the first round. Equivalently, we need to show that
experimenting gives a weakly lower cost than all agents playing safe for two
rounds ($2(S_0 + S_1 N)$), which leads to the constraint 
\begin{align}  \label{experiment}
\beta L + (1-\beta) H + \beta L x_L + (1-\beta) (S_0 + S_1(N-x_H)) \leq
2(S_0+S_1 N).
\end{align}

The CP then solves the constrained optimization problem given in %
\eqref{eq:opt2},
where the objective function is the total travel time summed over the two
periods and the IC constraints \eqref{2stage:ic1}, \eqref{2stage:ic2} and %
\eqref{2stage:ic3} can be explicitly rewritten as detailed in (\ref{agent2:S}%
), (\ref{agent3:R}), and (\ref{experiment}) respectively.

Finally it is easy to show that the choices $\pi_2(L)=\pi_2(H)=0$ and, for $%
\beta\ge\beta_f$, $\pi_2(L)=x^{\textup{eq}}_L-1, \pi_2(H)=0$ are feasible (i.e.
satisfy (\ref{agent2:S}), (\ref{agent3:R}), and (\ref{experiment})) and
lead to the same social cost as private and full information respectively,
thus proving that partial information is weakly better than private and full
information.
\end{enumerate}
\end{proof}

\section{Proofs of Section 4: Infinite horizon}

\subsection{Proof of Proposition \protect\ref{prop:so}}

We start by showing that under the given assumptions $x_H^{\textup{SO}} = 0$
and $x_L^{\textup{SO}} \geq 2$.

\begin{itemize}
\item $\boldsymbol{x_H^{\textup{$\mathbf{SO}$}} = 0}:$ Note that 
\begin{align*}
\mathbb{E} _{\theta_t}\left[ g(x, \theta_t) \mid \beta_{t-1}=0 \right]=\mu_H
x^2 +S_0 (N-x)= g(x,\mu_H).
\end{align*}
We next show that under the given assumptions $g(0,\mu_H) \leq g(1,\mu_H)$.
Together with the fact that $g(x,\mu_H)$ is strongly convex in $x$, this
proves the desired statement. Note that 
\begin{equation*}
g(0, \mu_H) \leq g(1, \mu_H) \Leftrightarrow S_0 N \leq \mu_H +S_0(N-1)
\Leftrightarrow S_0 \leq \mu_H, 
\end{equation*}
and the latter inequality holds by Assumption \ref{ass:mul_muh}.

\item $\boldsymbol{x_L^{\textup{\textbf{SO}}} \ge 2}:$ $x_L^{\textup{SO}}$
is the minimizer of $g(x, \mu_L)$, which is strongly convex in $x$. For the
minimizer to be $\ge 2$, the cost at $x=2$ must be strictly less than at $x=1
$, that is, 
\begin{equation*}
g(2, \mu_L)=4 \mu_L + S_0 (N-2) < \mu_L + S_0 (N-1)=g(1, \mu_L)
\end{equation*}
rearranging gives 
\begin{equation*}
\mu_L < \frac{1}{3} S_0, 
\end{equation*}
which holds by Assumption \ref{ass:mul_muh}.
\end{itemize}

Since the CP has full control, \eqref{so} is an optimal control problem and
we can apply the one-step deviation principle to prove optimality. To this
end, we distinguish two cases

\begin{itemize}
\item Consider any $\beta$ such that $\boldsymbol{x^{\textup{SO}}_{\beta}\ge 1}$:%
\newline
Sending a number of agents different from $x_{\beta}^{\textup{SO}}$ doesn't lead to a
profitable deviation. In fact, the stage cost would be higher (since $%
x_{\beta}^{\textup{SO}}$ is the minimizer of $g(x,\mu_{\beta})$) and the continuation
cost would be the same (as no more information can be gained by sending more
agents to the risky road).

\item Consider any $\beta$ such that $\boldsymbol{x^{\textup{SO}}_{\beta}= 0}$:\newline
In this case sending $\pi^{\textup{SO}}(\beta)=1$ agent is not myopically optimal and
the stage cost could be reduced by sending no agent. Nonetheless, we show
that because sending one agent provides information about the state of the
risky road, $\pi^{\textup{SO}}(\beta)=1$ is the best strategy if the CP is forward
looking. There are two possible deviations:

\begin{enumerate}
\item \underline{The CP sends more than one agent:} similarly to the
previous case the stage cost increases and the continuation cost stays the
same. Hence this deviation is not profitable.

\item \underline{The CP does not send any agents:} to analyze this case, let 
$\beta^{\prime }$ be the belief that $\theta_{t+1}=L$ when an agent has
belief $\beta$ that $\theta_{t}=L$ (i.e. $\beta^{\prime
}=\beta(1-\gamma_L)+(1-\beta)\gamma_H$). We start by noting that $%
x^{\textup{SO}}_{\beta} = 0 \Rightarrow x^{\textup{SO}}_{\beta^{\prime }}\le 1 \Rightarrow
\pi^{\textup{SO}}(\beta^{\prime })=1$ (see Lemma \ref{beta} below). The cost that the
CP encounters by sending one agent at time $t$ is 
\begin{align*}
V(\beta) = &\underbrace{\mu_{\beta}+S_0(N-1)}_\textup{stage cost time t} \\
&+\delta \underbrace{\left( \beta^{\prime }(\mu_L
(x^{\textup{SO}}_L)^2+S_0(N-x^{\textup{SO}}_L))+(1-\beta^{\prime })(\mu_H+S_0(N-1))\right)}_%
\textup{stage cost time t+1}+\delta^2\ldots
\end{align*}
while if he deviates and sends no agent at time $t$ the cost is 
\begin{align*}
\tilde V(\beta) &= \underbrace{S_0N}_\textup{stage cost time t}+\delta 
\underbrace{ (\mu_{\beta^{\prime }} \pi^{\textup{SO}}(\beta^{\prime})^2
+S_0(N-\pi^{\textup{SO}}(\beta^{\prime })))}_\textup{stage cost time t+1}%
+\delta^2\ldots \\
&= \underbrace{S_0N}_\textup{stage cost time t}+\delta \underbrace{\left(
\beta^{\prime }(\mu_L +S_0(N-1))+(1-\beta^{\prime })(\mu_H+S_0(N-1))\right)}_%
\textup{stage cost time t+1}+\delta^2\ldots
\end{align*}
where we used $\mu_{\beta^{\prime }}=\beta^{\prime }\mu_L+(1-\beta^{\prime
})\mu_H$ and $\pi^{\textup{SO}}(\beta^{\prime })=1$. Note that we did not report the
stage costs from time $t+2$ on as they are equal under both schemes.
Overall, $\pi^{\textup{SO}}(\beta)=1$ is optimal if $V(\beta)\le \tilde V(\beta)$ or
equivalently, 
\begin{align}
&\mu_{\beta}+S_0(N-1)+ \delta\beta^{\prime }(\mu_L
(x^{\textup{SO}}_L)^2+S_0(N-x^{\textup{SO}}_L)) \le S_0N + \delta \beta^{\prime }(\mu_L
+S_0(N-1))  \notag \\
&\Leftrightarrow \mu_{\beta} \le S_0 + \delta \beta^{\prime
}[g(1,\mu_L)-g(x^{\textup{SO}}_L,\mu_L)].  \label{condition_beta}
\end{align}
Note that $g(1,\mu_L)-g(x^{\textup{SO}}_L,\mu_L)\ge0$ since $x^{\textup{SO}}_L=\arg\min_x
g(c,\mu_L)$. Moreover when $\beta$ increases $\mu_{\beta}$ decreases. Hence it
suffices to prove that \eqref{condition_beta} holds for the smallest
possible value of $\beta$ which is $0$.

We note that $g(1,\mu_L)-g(x^{\textup{SO}}_L,\mu_L) \ge S_0-3\mu_L$ since 
\begin{align*}
g(1,\mu_L)-g(x^{\textup{SO}}_L,\mu_L) &\ge g(1,\mu_L)-g(2,\mu_L) \\
&= \mu_L +S_0(N-1)- 4\mu_L -S_0(N-2)=S_0-3\mu_L.
\end{align*}
Note that $S_0 - 3\mu_L \geq \frac{S_0}{3} - \mu_L\ge 0$ by Assumption \ref%
{ass:mul_muh} and thus a sufficient condition for \eqref{condition_beta} to
hold when $\beta=0$ (and $\beta'=\gamma_H$) is 
\begin{equation*}
\mu_H \le S_0 + \delta \gamma_H \left[ \frac{S_0}{3}-\mu_L\right], 
\end{equation*}
which holds by Assumption \ref{ass:mul_muh}.
\end{enumerate}
\end{itemize}

\begin{lemma}
\label{beta} $x^{\textup{SO}}_{\beta} = 0 \Rightarrow x^{\textup{SO}}_{\beta^{\prime }}\le 1$.
\end{lemma}

\begin{proof}
A sufficient condition for $x^{\textup{SO}}_{\beta^{\prime }}\le 1$ is 
\begin{equation*}
g(1,\mu_{\beta^{\prime }})= \mu_{\beta^{\prime }}+ S_0(N-1) <
4\mu_{\beta^{\prime }} + S_0(N-2) = g(2,\mu_{\beta^{\prime }})
\Leftrightarrow \mu_{\beta^{\prime }}> \frac{S_0}{3}.
\end{equation*}
We next show that $\mu_{\beta^{\prime }}> \frac{\mu_{\beta}}{3}$. The
conclusion then follows since $x^{\textup{SO}}_{\beta} = 0$ implies 
\begin{equation*}
g(0,\mu_{\beta})= S_0N<\mu_{\beta}+ S_0(N-1)= g(1,\mu_{\beta}) \Leftrightarrow
\mu_{\beta}>S_0.
\end{equation*}
To show $3\mu_{\beta^{\prime }}> {\mu_{\beta}}$ recall that $%
\mu_{\beta^{\prime }}=\beta\mu_L+(1-\beta)\mu_H\ge \beta L+(1-\beta)\mu_H$
and $\mu_H=(1-\gamma_H)H+\gamma_H L \ge \frac12H$. Hence 
\begin{align*}
3\mu_{\beta^{\prime }}&\ge 3\beta L+3(1-\beta)\mu_H \ge \frac32\beta
L+\frac32\beta L+\frac32(1-\beta)H \ge\frac32\mu_{\beta} >\mu_{\beta}.
\end{align*}
\end{proof}

\subsection{Preliminary statements in support of the proof of Proposition 
\protect\ref{prop:eq}}

To prove our main Proposition \ref{prop:eq} we start with some additional
statements. We first prove that the agent's state can be simplified as
detailed in Lemma \ref{lemma:agent_state} in the main text.\newline

\noindent \textbf{Proof of Lemma \ref{lemma:agent_state}:} If all agents are
following the scheme $\pi_{c,d}$ then the flow on the risky road at time $t-1
$ is distinct depending on whether $\theta_{t-2}=H$ or $\theta_{t-2}=L$.
Thus, either an agent was on the risky road at time $t-1$ and observed $%
\theta_{t-1}$ or the agent was on safe and can infer $x_{t-1}$ and (from
that) $\theta_{t-2}$. By the Markov property of $\theta$ and the
stationarity of the recommendation policy no information before $\theta_{t-2}
$ is useful to the agents. Thus, the state $z_t^i$ is a sufficient summary
for any agent to determine his expected ongoing cost, as well as the
information that other agents have. If agents do not follow the scheme, then
all agents receive full information and this state is still sufficient as
every agent and the CP will have symmetric information. \hfill{$\blacksquare$%
}\newline

\noindent From here on we consider the values of $c,d$ fixed (satisfying the
assumptions of \hbox{Proposition \ref{prop:eq}}) and we denote by $u^*(z^i_t)$ the
expected cost under $\pi_{c,d}$ of an agent whose state is $z^i_t$. We note
that the expected cost for any agent that knows that the risky road was high
at the previous step ($\theta_{t-1} = H$) and before receiving a
recommendation for time $t$ is the same, no matter his state. Intuitively,
this is true because according to the recommendation scheme $\pi_{c,d}$ if $%
\theta_{t-1} = H$ then at the next step the CP sends the recommendation $r_R$
to one and only one agent (the experimenter) selected at random among all
the agents independent of previous actions or knowledge.

\begin{lemma}
\label{v} The expected cost 
\begin{equation*}
\bar{v} := \mathbb{E} _{\pi_{c,d}} [ u^*([x_{t-1}, \theta^i_{t-1},
r^i_{t-1}]) \mid x_{t-1}=x,\theta^i_{t-1}=H] 
\end{equation*}
is the same for all $x\ge 1$.
\end{lemma}

\begin{proof}
Whenever $\theta_{t-1} = H$ is observed, a new experimenter is chosen among
all agents. Thus every agent, no matter which road he was on, has an
identical likelihood of being chosen as the experimenter at time $t$. Note
that $\bar{v}$ conditions on the knowledge that $\theta_{t-1} = H$, hence
there is no need for distinguishing states where agents do not know the
state of the road. In other words, the expectation is only over the
recommendation scheme $\pi_{c,d}$ and thus 
\begin{align*}
\bar{v} &= \frac{1}{N} \underbrace{u^*([x_{t-1}, H, r_R])}_{\text{ongoing
cost as the experimenter}} + \frac{N-1}{N} \underbrace{u^*([x_{t-1}, H, r_S])%
}_{\text{ongoing cost on safe}} \\
&= \frac{1}{N} u^*([1, H, r_R]) + \frac{N-1}{N} u^*([1, H, r_S]).
\end{align*}
In the second line we substitute the observed flow with $1$ since it is
unimportant (i.e. these costs are the same for any $x_{t-1} \geq 1$); given
the state $\theta_{t-1} = H$, the flow in the previous round will not effect
the ongoing cost: as at time $t$ one
agent will be on risky and $N-1$ will be on safe.
\end{proof}

Lemma \ref{v} simplifies our analysis because it implies that we can
partition the infinite horizon into consecutive periods by defining the
beginning of a new period as the time immediately after the risky road
switches from $L$ to $H$, see Table \ref{fig:scheme_so} in the main text.
Conditioned on the agents knowing that a new period has begun, their ongoing
cost from that point on is the same (i.e. $\bar v$) independent of their
history. This observation simplifies the analysis of incentive compatibility
and optimality. Table \ref{fig:scheme_so} in the main text illustrates the
flow in the risky road under the social optimum and the scheme described in
Proposition \ref{prop:eq} within one period. Table \ref{fig:scheme_pistar}
illustrates the expected stage cost for agents taking the risky road or the
safe road under the scheme described in Proposition \ref{prop:eq}.

\begin{table}[]
\caption{Expected stage costs if agents follow the recommendation scheme
given in Proposition \protect\ref{prop:eq}}
\label{fig:scheme_pistar}%
\begin{adjustwidth}{-70pt}{-70pt}
\begin{center}
\begin{footnotesize}
\begin{tabular}{|c|c|ccc|c|ccc| }\hline
State of the risky road& H& H &$\hdots$&  L& L& L & $\hdots$&  H\\[0.1cm] \hline
Number of agents on risky & -& 1 &$\hdots$& 1& $c$& $d$ & $\hdots$& $d$ \\[0.1cm] \hline
Expected stage cost on risky & -& $\mu_H$ &$\hdots$& $\mu_H$ & $\mu_L c$& $\mu_L d$ & $\hdots$& $\mu_L d$ \\[0.1cm] \hline
Expected stage cost on safe& -& $S_0$ &$\hdots$& $S_0$ & $S_0$& $S_0$ & $\hdots$& $S_0$ \\[0.1cm] \hline
\end{tabular}
\end{footnotesize}
\end{center}
\end{adjustwidth}
\end{table}

We start our analysis by deriving closed form expressions and relations for
the cost $u^*(z)$ of different states $z$ reached under $\pi_{c,d}$.

\begin{lemma}[Closed form expression of auxiliary cost]
\ \label{lemma:closed}

\begin{enumerate}
\item $u^*([c, L, r_R]) = u^*([d, L, r_R])= \frac{1}{1-\delta (1-\gamma_L)}%
\left( \mu_Ld+\delta \gamma_L \bar v\right)$

\item $u^*([c, L, r_S]) = u^*([d, L, r_S])=\frac{1}{1-\delta(1-\gamma_L)}
(S_0 + \delta \gamma_L \bar{v})$
\end{enumerate}
\end{lemma}

\begin{proof}
\begin{enumerate}
\item The first equality follows from the fact that the flow on the risky
road after $\theta = L$ and the flow $c$ or $d$ is $d$. The second equality
follows from $u^*([d, L, r_R])= \mu_Ld+\delta((1-\gamma_L) u^*([d, L,
r_R])+\gamma_L \bar v)$ (see \eqref{xll_L}).

\item The first equality follows similarly to the above point and the second
follows from $u^*([d, L, r_S])=S_0 +\delta((1-\gamma_L )u^*([d, L,
r_S])+\gamma_L \bar v)$ (see \eqref{xll_L}).
\end{enumerate}
\end{proof}

\begin{lemma}
\label{lemma:mlxl} $\mu_L c \leq \mu_L d \leq S_0$ and $c\ge \frac{S_0}{%
2\mu_L}-\frac{1}{2}$.
\end{lemma}

\begin{proof}
The first chain of inequalities follows immediately from the assumption $c \leq d\le x^{\textup{eq}}_L$ and $\mu_L x_L^{\textup{eq}} \leq S_0$ by the equilibrium condition. Finally, since $x_L^{\textup{SO}}$ is either the closest
integer that to $\frac{S_0}{2\mu_L}$, one gets $c\ge
x_L^{\textup{SO}}\ge \frac{S_0}{2\mu_L}-\frac{1}{2}$.
\end{proof}

\begin{remark}
\label{rmrk:full_info} According to Definition \ref{def:am}, if any agent deviates the punishment is full
information. Specifically, we assume that the CP sends a recommendation of
risky to each agent with probability $\frac{S_0}{N\mu_{\beta}}$, so that the
expected cost on risky is exactly equal to the  fixed cost $S_0$ of the safe
road. The continuation cost after any deviation is therefore 
\begin{equation*}
u^*_{dev}=S_0+\delta S_0+\delta^2S_0+\hdots =\frac{1}{1-\delta} S_0,
\end{equation*}
Independent of the belief $\beta$.
\end{remark}

\begin{lemma}
\label{lemma:punish} The following statements hold:

\begin{enumerate}
\item $\bar{v} < \frac{1}{1-\delta}S_0$

\item $u^*([d, L, r_R]) \leq \frac{1}{1-\delta} S_0$

\item $u^*(1,L) :=\left( \frac{c-1 }{N-1}u^*([1,L,r_R]) +\frac{N-c }{N-1}%
u^*([1,L,r_S]) \right) \leq u^*([d, L, r_S]) \leq \frac{1}{1-\delta}S_0$

\item $u^*(c, L) :=\left( \frac{d - c}{N-c} u^*([c,L,r_R]) + \frac{N-d }{N-c}%
u^*([c,L,r_S])\right) \leq u^*([d, L, r_S]) \leq \frac{1}{1-\delta} S_0$
\end{enumerate}
\end{lemma}

\begin{proof}
\begin{enumerate}
\item 
\begin{align*}
\bar{v} &= \frac{1}{N} \mu_H + \frac{N-1}{N} S_0 + \delta \left( \gamma_H
\left( \frac{c}{N} \mu_L c + \frac{N-c}{N} S_0 \right. \right. \\
&\qquad \left. \left. + \delta \left( \gamma_L \bar{v} + (1-\gamma_L) \left(%
\frac{d}{N} u^*([d, L, r_R]) + \frac{N-d}{N} u^*([d, L, r_S]) \right)
\right) \right) + (1-\gamma_H) \bar{v} \right)
\end{align*}
where $u^*([d, L, r_R])$ and $u^*([d, L, r_S])$ are as defined in Lemma \ref%
{lemma:closed}. To simplify exposition recall that 
\begin{align*}
g(1, \mu_H) &:= \mu_H + (N-1) S_0 \\
g(c,\mu_L) &:= \mu_L c^2 + (N-c) S_0 \\
g(d,\mu_L) &:=\mu_L d^2 + (N-d) S_0
\end{align*}
then $\bar{v}$ can be written as 
\begin{align}  \label{v_bar}
\bar{v} = \frac{1-\delta(1-\gamma_L)}{N(1-\delta)(1-\delta(1-\gamma_H-%
\gamma_L))} \left( g(1, \mu_H) + \delta \gamma_H \left( g(c,\mu_L) + 
\frac{\delta(1-\gamma_L)}{1-\delta(1-\gamma_L)} g(d,\mu_L) \right)\right)
\end{align}
To show the inequality holds we first multiply both sides by $(1-\delta) N$ 
\begin{align*}
&\frac{1-\delta(1-\gamma_L)}{1-\delta(1-\gamma_L-\gamma_H)} \left( g(1, \mu_H) + \delta \gamma_H \left( g(c,\mu_L) + \frac{\delta (1-\gamma_L)}{%
1-\delta(1-\gamma_L)} g(d,\mu_L) \right) \right) < S_0 N \\
\iff &(1-\delta(1-\gamma_L))( g(1, \mu_H) + \delta \gamma_H g(c,\mu_L))
+ \delta^2 \gamma_H (1-\gamma_L) g(d,\mu_L) \\
&\qquad \qquad < (1-\delta(1-\gamma_L-\gamma_H)) S_0 N \\
\iff &(1-\delta(1-\gamma_L))(\mu_H + S_0 (N-1) + \delta \gamma_H g(c,\mu_L))
+ \delta^2 \gamma_H (1-\gamma_L) g(d,\mu_L) \\
&\qquad \qquad < (1-\delta(1-\gamma_L)) S_0 N + \delta \gamma_H S_0 N -
\delta^2 \gamma_H (1-\gamma_L) S_0 N + \delta^2 \gamma_H (1-\gamma_L) S_0 N
\\
\iff &(1-\delta(1-\gamma_L))(\mu_H + \delta \gamma_H g(c,\mu_L)) + \delta^2
\gamma_H (1-\gamma_L) g(d,\mu_L) \\
&\qquad \qquad < (1-\delta(1-\gamma_L)) (S_0 + \delta \gamma_H S_0 N) +
\delta^2 \gamma_H (1-\gamma_L) S_0 N.
\end{align*}
By Lemma \ref{lemma:mlxl}, $g(d,\mu_L) \leq S_0 N$, so it is sufficient to
show 
\begin{align*}
&\mu_H + \delta \gamma_H g(c,\mu_L) < S_0 + \delta \gamma_H S_0 N \iff \mu_H
< S_0 + \delta \gamma_H (S_0 N - g(c,\mu_L)).
\end{align*}
Note that, by assumption, $g(c,\mu_L) \le g(2,\mu_L)< g(1,\mu_L)$ hence 
\begin{equation*}
S_0 N - g(c,\mu_L) > S_0 N - (S_0(N-1) + \mu_L) = S_0- \mu_L > \frac{S_0}{3}%
- \mu_L . 
\end{equation*}
Thus, it is sufficient if 
\begin{equation*}
\mu_H \leq S_0 + \delta \gamma_H \left( \frac{S_0}{3} - \mu_L \right)
\end{equation*}
which holds by Assumption \ref{ass:mul_muh}.

\item The cost $u^*([d, L, r_R])$ is defined in Lemma \ref{lemma:closed}.
Then 
\begin{align*}
u^*([d,L, r_R])= &\frac{1}{1-\delta(1-\gamma_L)} (\mu_L d + \delta \gamma_L \bar{v}) \leq 
\frac{1}{1-\delta} S_0 \\
\iff &(1-\delta)(\mu_L d + \delta \gamma_L \bar{v}) \leq
(1-\delta(1-\gamma_L)) S_0 \\
\iff &(1-\delta)\mu_L d + (1-\delta)\delta \gamma_L \bar{v} \leq (1-\delta)
S_0 + \delta \gamma_L S_0
\end{align*}
and the result holds by part 1 of this lemma and by Lemma \ref{lemma:mlxl}.

\item Expanding the cost $u^*(1, L)$ 
\begin{align*}
u^*(1,L) &= \frac{c-1}{N-1} \mu_L c + \frac{N-c}{N-1} S_0 \\
&\qquad \qquad + \delta \left( (1-\gamma_L) \left( \frac{d-1}{N-1} u^*([d,
L, r_R]) + \frac{N-d}{N-1} u^*([d, L, r_S]) \right) + \gamma_L \bar{v}
\right).
\end{align*}
Note that $u^*([d, L, r_S])=S_0+\delta(1-\gamma_L)u^*([d, L,
r_S])+\delta\gamma_L \bar v.$ Then $u^*(1,L)\le u^*([d,L,r_S])$ holds if 
\begin{align*}
\frac{c-1}{N-1} \mu_L c + \frac{N-c}{N-1} S_0 + \frac{\delta(1-\gamma_L)}{%
1-\delta(1-\gamma_L)} \left( \frac{d-1}{N-1} \mu_L d + \frac{N-d}{N-1} S_0
\right) \leq S_0 + \frac{\delta(1-\gamma_L)}{1-\delta(1-\gamma_L)} S_0
\end{align*}
which is true by Lemma \ref{lemma:mlxl}.

For the second inequality, we need 
\begin{align}  \label{eq:so_v}
&\frac{1}{1-\delta(1-\gamma_L)} (S_0 + \delta \gamma_L \bar{v}) \leq \frac{1%
}{1-\delta} S_0 \\
\iff &(1-\delta)(S_0 + \delta \gamma_L \bar{v}) \leq (1-\delta+\delta
\gamma_L) S_0  \notag
\end{align}
and the result follows from the first point of the lemma.

\item By using Lemma \ref{lemma:closed}, we can expand $u^*(c, L)$ as 
\begin{align*}
u^*(c, L) &= \frac{d-c}{N-c} \left( \frac{1}{1-\delta (1-\gamma_L)}\left(
\mu_Ld+\delta \gamma_L \bar v\right) \right)+ \frac{N-d}{N-c} \left( \frac{1%
}{1-\delta (1-\gamma_L)}\left( S_0+\delta \gamma_L \bar v\right) \right) \\
&= \frac{1}{1-\delta(1-\gamma_L)} \left( \frac{d-c}{N-c} \mu_L d + \frac{N-d%
}{N-c} S_0 + \delta \gamma_L \bar{v} \right),
\end{align*}
It follows by Lemma \ref{lemma:mlxl}, that
\begin{equation*}
u^*(c, L) \leq \frac{1}{1-\delta(1-\gamma_L)} (S_0 + \delta \gamma_L \bar{v}%
)=u^*([d,L,r_S]) 
\end{equation*}
and the result follows from part 3 of this lemma.
\end{enumerate}
\end{proof}

\subsection{Proof of Proposition \protect\ref{prop:eq}}

\label{proof2}

As in the two stage case, we break the agents into \textbf{types based on
individual's information.} We detail the possible states agents reach under $%
\pi_{c,d}$ and show that $\xi_{\pi_{c,d}}$ is an equilibrium by showing that
at each state any deviation would lead to a higher expected cost. Each
equilibrium constraint in this setting is dynamic (i.e. consists of both
stage and continuation cost).

\subsubsection*{Type 1 ($\protect\beta_{t-1}^i = L$, $r_{t-1}^i = r_R$)}

Agents of this type took the risky road at time $t-1$ and observed L. Under $%
\pi_{c,d}$ these agents receive a recommendation to remain on risky, so $%
r_{t-1}^i = r_R$ for all cases below. We further distinguish different
states based on the observed flow on risky at time $t-1$. We present our
results in decreasing order of number of other agents that took risky at $t-1
$. We show that in each case an agent that observes $L$ follows the
recommendation and stays on risky.

\begin{itemize}
\item {$\mathbf{[x_{t-1},\beta^i_{t-1},r^i_{t-1}]=[d, L, r_R]}$:} this agent
is part of what we call the ``incentive compatible'' flow, which is the flow 
$d$. If everybody follows the recommendation then $x_t=d$. Agent $i$'s
expected costs under ${\pi_{c,d}}$ (if he follows or if he deviates) are
therefore

\begin{itemize}
\item[-] \underline{Following:} 
\begin{align}  \label{xll_L}
u^*([d, L, r_R]) &= \underbrace{\mu_L d}_{\substack{\textup{stage cost}\\ \textup{for $t$}}} +
\delta \big( (1-\gamma_L)\underbrace{ u^*([d, L, r_R])}_{\substack{ \textup{ongoing
cost} \\ \textup{ if $\theta_t=L$}}} + \gamma_L \underbrace{\mathbb{E} _{
\pi_{c,d}}[u^*([d, H, r^i])] }_{\textup{ongoing cost if $\theta_t=H$}} \big)%
,  \notag \\
&= \mu_L d + \delta \left( (1-\gamma_L){\ u^*([d, L, r_R])} + \gamma_L {\bar{v%
}} \right),
\end{align}

\item[-] \underline{Deviating to safe:} 
\begin{align}  \label{dev_xll_L}
\underbrace{S_0}_{\textup{stage cost for $t$}} + \underbrace{\delta
u^*_{dev}}_{\textup{ongoing cost}}=\ S_0+ \frac{\delta}{1-\delta} S_0
\end{align}
\end{itemize}

Thus, taking the risky road is an equilibrium if 
\begin{align}  \label{IC:xll_L}
&\mu_L d + \delta \left( (1-\gamma_L){\ u^*([d, L, r_R])} + \gamma_L {\bar{v}%
} \right) \leq S_0 + \frac{\delta}{1-\delta} S_0.
\end{align}

The inequality holds by Lemmas \ref{lemma:mlxl}, \ref{lemma:punish}-1, and %
\ref{lemma:punish}-2.

\item $\mathbf{[x_{t-1},\beta^i_{t-1},r^i_{t-1}]=[c, L, r_R]}$ this agent is
part of the ``exploiters" (agents that are sent to the risky road the first
period after the experimenter saw L). If everybody follows the
recommendation the flow on risky will be $x_t = d$.

By Lemma \ref{lemma:closed}, the costs are the same as in state $[d, L, r_R]$
(as given in (\ref{xll_L}) and (\ref{dev_xll_L})). The equilibrium
constraint is therefore identical to (\ref{IC:xll_L}) and is satisfied.

\item $\mathbf{[x_{t-1},\beta^i_{t-1},r^i_{t-1}]=[1, L, r_R]}$ this agent is
the current experimenter and just saw the road change to ``low". If
everybody follows the recommendation the next flow on risky would be $x_t = c
$. His expected costs are

\begin{itemize}
\item[-] \underline{Following:} 
\begin{align}  \label{eq:cost_1_L_empty}
u^*([1, L, r_R]) &={\ \mu_L c}+ \delta ((1-\gamma_L) {u^*([c, L, r_R])} +
\gamma_L \bar{v}).
\end{align}

\item[-] \underline{Deviating to safe:} 
\begin{align*}
{S_0} + \delta {u^*_{\textup{dev}}} = S_0+ \frac{\delta}{1-\delta} S_0 
\end{align*}
\end{itemize}

The agent follows the recommendation if 
\begin{equation*}
\mu_L c + \delta( (1-\gamma_L) u^*([c, L, r_R]) + \gamma_L \bar{v}) \leq S_0
+ \frac{\delta}{1-\delta} S_0. 
\end{equation*}

This inequality holds by Lemmas \ref{lemma:closed}, \ref{lemma:mlxl}, \ref%
{lemma:punish}-1, and \ref{lemma:punish}-2.
\end{itemize}

\subsubsection*{Type 2 ($\protect\beta_{t-1}^i = H$)}

These agents took risky and saw $H$, this means a new experimenter will be
chosen. They each will receive a recommendation of either $r_R$ or $r_S$.
Since the state is $H$ the flow on risky in the next round is $x_t=1$ no
matter the previous flow and we can divide the states by recommendations

\begin{itemize}
\item $r^i_{t-1}=r_S$. He is not the experimenter.

\begin{itemize}
\item[-] \underline{Following the recommendation and taking safe} 
\begin{align}
u^*([-, H, r_S]) = S_0 + \delta(\gamma_H u^*(1,L) + (1-\gamma_H) \bar{v})
\end{align}

\item[-] \underline{Deviating} 
\begin{align*}
2 \mu_H + \frac{\delta}{1-\delta} S_0
\end{align*}
\end{itemize}

Thus, the incentive constraint holds by Lemma \ref{lemma:punish}-1, \ref%
{lemma:punish}-3, and Assumption \ref{ass:mul_muh}.

\item \textbf{$r^i_{t-1}=r_R$}. This agent has been selected to be the
experimenter. Recall that if the agent deviates then full information is
provided from then on and the cost is $S_0$ at every round (see Remark \ref%
{rmrk:full_info}). The agent's expected costs are

\begin{itemize}
\item[-] \underline{Following the recommendation and taking risky} 
\begin{align}  \label{eq:cost_1_h_r}
u^*([x_{t-1}, H, r_R]) = \mu_H + \delta(\gamma_H u^*([1,L, r_R]) +
(1-\gamma_H) \bar{v})
\end{align}

\item[-] \underline{Deviating} 
\begin{align*}
S_0 + \frac{\delta}{1-\delta} S_0
\end{align*}
\end{itemize}

The incentive constraint can be written 
\begin{align*}
\mu_H + \delta (\gamma_H u^*([1, L, r_R]) + (1-\gamma_H) \bar{v}) \leq S_0 +
\gamma_H \frac{\delta}{1-\delta} S_0 + (1-\gamma_H) \frac{\delta}{1-\delta}
S_0
\end{align*}
and by Lemma \ref{lemma:punish}-1 it suffices to show 
\begin{align*}
&\mu_H + \delta \gamma_H u^*([1, L, r_R]) \leq S_0 + \gamma_H \frac{\delta}{%
1-\delta} S_0.
\end{align*}
Note that 
\begin{align*}
u^*([1, L, r_R]) = \mu_L c + \delta ( (1-\gamma_L) u^*([c, L, r_R]) +
\gamma_L \bar{v} ) \leq \mu_L c + \frac{\delta}{1-\delta} S_0
\end{align*}
where the upper bound follows from Lemma \ref{lemma:closed}-1, \ref%
{lemma:punish}-1 and \ref{lemma:punish}-2. Thus, it is sufficient if 
\begin{align*}
\mu_H &\leq S_0 + \delta \gamma_H (S_0 - \mu_L c) \\
&\leq S_0 + \delta \gamma_H \left( S_0 - \mu_L \left( \frac{S_0}{2 \mu_L} - 
\frac{1}{2} \right) \right) \\
&= S_0 + \delta \gamma_H \frac{1}{2}\left(S_0 + \mu_L \right)
\end{align*}
where the second inequality follows by Lemma \ref{lemma:mlxl}. The result
holds by \hbox{Assumption \ref{ass:mul_muh}.}
\end{itemize}

\subsubsection*{Type 3 ($\protect\beta_{t-1}^i = U, r_{t-1}^i = r_S$)}

These agents took safe at time $t-1$ and received a recommendation to remain
on safe, $r_S$. Note that since these agents took safe at time $t-1$ they do
not know $\theta_{t-1}$. Receiving a recommendation of safe either means
that $\theta_{t-1} = L$ (and the agent continues to be part of the flow on
the safe road) or that $\theta_{t-1} = H$ (a new cycle has begun but the
agent is not the new experimenter). For simplicity we denote the probability
of the first event by $p_{x,S}=\mathbb{P}(\theta_{t-1}=L\mid x_{t-1}=x,
r_{t-1}^i = r_S)$. Note that this probability depends on the flow $x$
observed at $t-1$. The following lemma relates these probabilities for the
cases $x_{t-1}=d$ and $x_{t-1}=c$.

\begin{lemma}
\label{lem:prob} The following statements hold

\begin{enumerate}
\item $p_{d,S}=\frac{1-\gamma_L}{1-\gamma_L+ \gamma_L \frac{N-1}{N}}$

\item $p_{c,S} =\frac{(1-\gamma_L) \frac{N-d}{N-c}}{(1-\gamma_L) \frac{N-d}{%
N-c} + \gamma_L \frac{N-1}{N} }$

\item $p_{1,S}=\frac{\gamma_H \frac{N-c}{N-1}}{(1-\gamma_H) \frac{N-1}{N} +
\gamma_H \frac{N-c}{N-1}}$

\item $p_{c,S} \le p_{d,S}$
\end{enumerate}
\end{lemma}

\begin{proof}
\begin{enumerate}
\item Note that the agent can infer $\theta_{t-2} = L$ from the flow being $d
$ on the risky road at time $t-1$. Therefore, by Bayes rule 
\begin{align}  \label{eq:def_tildep}
p_{d,S} &:= \mathbb{P}(\theta_{t-1} = L \mid [d, U, r_S]) \\
&= \frac{(1-\gamma_L) \mathbb{P}(r_S \mid \theta_{t-1} = L, x_{t-1} = d)}{%
(1-\gamma_L) \mathbb{P}(r_S \mid \theta_{t-1} = L, x_{t-1} = d) + \gamma_L 
\mathbb{P}(r_S \mid \theta_{t-1} = H, x_{t-1} = d)}  \notag \\
&= \frac{1-\gamma_L}{1-\gamma_L+ \gamma_L \frac{N-1}{N}}.
\end{align}

\item Note that the agent can infer $\theta_{t-2} = L$ from the flow being $c
$ on the risky road at time $t-1$. By Bayes rule 
\begin{align}  \label{eq:def_hatp}
p_{c,S} &:= \mathbb{P}(\theta_{t-1} = L \mid [c, U, r_S]) \\
&= \frac{(1-\gamma_L) \mathbb{P}(r_S \mid \theta_{t-1} = L, x_{t-1} = c)}{%
(1-\gamma_L) \mathbb{P}(r_S \mid \theta_{t-1} = L, x_{t-1} = c) + \gamma_L 
\mathbb{P}(r_S \mid \theta_{t-1} = H, x_{t-1} = c)}  \notag \\
&= \frac{(1-\gamma_L) \frac{N-d}{N-c}}{(1-\gamma_L) \frac{N-d}{N-c} +
\gamma_L \frac{N-1}{N} }.  \notag
\end{align}

\item Note that the agent can infer $\theta_{t-2} = H$ from the flow being $1
$ on the risky road at time $t-1$. By Bayes rule 
\begin{align*}
p_{1,S} &:= \mathbb{P}(\theta_{t-1} = L \mid [1, U, r_S]) \\
&= \frac{\gamma_H \mathbb{P}(r_S \mid \theta_{t-1} = L, x_{t-1} = 1)}{%
(1-\gamma_H) \mathbb{P}(r_S \mid \theta_{t-1} = H, x_{t-1} = 1) + \gamma_H 
\mathbb{P}(r_S \mid \theta_{t-1} = L, x_{t-1} = 1)} \\
&= \frac{\gamma_H \frac{N-c}{N-1}}{\gamma_H \frac{N-c}{N-1} + (1-\gamma_H) 
\frac{N-1}{N}}.
\end{align*}

\item For positive $\alpha, \beta, \eta$ 
\begin{equation*}
\frac{\alpha}{\alpha+\beta} \geq \frac{\eta}{\eta+\beta} \iff \alpha \geq
\eta. 
\end{equation*}
Let $\alpha := 1-\gamma_L$, $\beta:= \gamma_L\frac{N-1}{N}$, and $\eta :=
(1-\gamma_L)\frac{N-d}{N-c}$. Then $p_{d,S} = \frac{\alpha}{\alpha+\beta}$
and $p_{c,S} = \frac{\eta}{\eta+\beta}$. The fact that $d \geq c$ implies $%
\alpha\geq \eta$ and therefore $p_{d,S}\ge p_{c,S} $.
\end{enumerate}
\end{proof}

The possible states for agents of Type 2 are

\begin{itemize}
\item $\mathbf{[x_{t-1},\beta^i_{t-1},r^i_{t-1}]=[d, U, r_S]}$: these agents
were on safe at time $t-1$, observed flow $d$ and received a recommendation
to remain on safe. Their expected costs are 

\begin{itemize}
\item[-] \underline{Following the recommendation and taking safe} 
\begin{align*}
u^*([d, U, r_S]) = p_{d,S} \underbrace{u^*([d,L,r_S])}_{\textup{cost if }
\theta_{t-1}=L} + (1-p_{d,S}) \underbrace{u^*([-,H,r_S])}_{\textup{cost if }
\theta_{t-1}=H}.
\end{align*}

\item[-] \underline{Deviating to risky} 
\begin{align*}
p_{d,S} \mu_L (d+1) + (1-p_{d,S}) 2 \mu_H + \frac{\delta}{1-\delta} S_0
\end{align*}
\end{itemize}

This inequality holds by assumption \eqref{eq:def_xll_simple}.

\item $\mathbf{[x_{t-1},\beta^i_{t-1},r^i_{t-1}]=[c, U, r_S]}$: The agent's
expected costs are

\begin{itemize}
\item[-] \underline{Following the recommendation and taking safe} 
\begin{align*}
u^*([c,U, r_S]) &= p_{c,S} \underbrace{u^*([d,L,r_S])}_{F_1} + (1-p_{c,S}) 
\underbrace{u^*([-,H,r_S])}_{F_2}.
\end{align*}

\item[-] \underline{Deviating} 
\begin{align*}
&p_{c,S} \underbrace{\left( \mu_L (d+1) + \frac{\delta}{1-\delta} S_0 \right)%
}_{D_1} + (1-p_{c,S}) \underbrace{\left( 2 \mu_H + \frac{\delta}{1-\delta}
S_0 \right)}_{D_2}
\end{align*}
\end{itemize}

By inspection, this case is similar to the previous case ($[d, U, r_S]$).
The only difference is the beliefs ($p_{d,S}$ for the previous case, and $%
p_{c,S}$ for this case). The incentive compatibility constraint %
\eqref{eq:def_xll_simple} for the previous case can be written compactly as 
\begin{equation*}
p_{d,S} F_1+ (1-p_{d,S}) F_2\le p_{d,S} D_1+ (1-p_{d,S}) D_2
\end{equation*}
or equivalently 
\begin{equation*}
p_{d,S} (F_1-D_1)+ (1-p_{d,S}) (F_2-D_2)\le 0. 
\end{equation*}
We next show that $(F_2-D_2)\le(F_1-D_1)$. Since, by Lemma \ref{lem:prob} $%
p_{c,S} \leq p_{d,S}$, by properties of convex combinations, this suffices to
show that 
\begin{equation*}
p_{c,S} (F_1-D_1)+ (1-p_{c,S}) (F_2-D_2)\le p_{d,S} (F_1-D_1)+ (1-p_{d,S})
(F_2-D_2)\le 0,
\end{equation*}
as desired. To show $(F_2-D_2)\le(F_1-D_1)$ we equivalently show $%
(F_2+D_1)\le(F_1+D_2)$. Note 
\begin{align*}
F_2 &= u^*([-,H,r_S]) = S_0 + \delta\left(\gamma_H u^*(1,L)+(1-\gamma_H)\bar
v\right) \\
F_1 &= u^*([d,L,r_S]) = S_0 + \delta((1-\gamma_L )u^*([d, L, r_S])+\gamma_L
\bar v).
\end{align*}
Hence $(F_2+D_1)\le(F_1+D_2)$ can be rewritten as 
\begin{align*}
&S_0 +\delta ( (1-\gamma_H) \bar{v} + \gamma_H u^*(1,L)) + \mu_L (d+1) + 
\frac{\delta}{1-\delta} S_0 \\
&\le S_0 + \delta( \gamma_L \bar{v} + (1-\gamma_L )u^*([d, L, r_S])) + 2
\mu_H + \frac{\delta}{1-\delta} S_0 \\
\iff & \delta(1-\gamma_L-\gamma_H) \bar{v} + \mu_L (d+1) \leq \delta
((1-\gamma_L) u^*([d, L, r_S]) - \gamma_H u^*(1,L)) + 2 \mu_H.
\end{align*}
By Lemma \ref{lemma:mlxl} and by Assumption \ref{ass:mul_muh}, $\mu_L d +
\mu_L \leq S_0 + S_0 \leq 2 \mu_H$ and it is sufficient to show 
\begin{align*}
(1-\gamma_L-\gamma_H) \bar{v} \leq (1-\gamma_L) u^*([d, L, r_S]) - \gamma_H
u^*(1,L).
\end{align*}
The right hand side can be lower bounded using Lemma \ref{lemma:punish}-3 by 
\begin{equation*}
(1-\gamma_L-\gamma_H) u^*([d, L, r_S]) 
\end{equation*}
and incentive compatibility holds if 
\begin{align*}
&\bar{v} \leq u^*([d, L, r_S]) \\
\iff &\bar{v} \leq \frac{1}{1-\delta(1-\gamma_L)} (S_0 + \delta \gamma_L 
\bar{v}) \\
\iff &(1-\delta) \bar{v} \leq S_0
\end{align*}
which holds by Lemma \ref{lemma:punish}-1.

\item {$\mathbf{[x_{t-1},\beta^i_{t-1},r^i_{t-1}]=[1, U, r_S]}$} These
agents were on safe at time $t-1$ and observed $x_{t-1} =1$, thus they infer
that $\theta_{t-2} = H$. The agent's expected costs are

\begin{itemize}
\item[-] \underline{Following the recommendation of safe} 
\begin{align*}
u^*([1, U, r_S]) &= p_{1,S} u^*([1,L,r_S])+ (1-p_{1,S}) u^*([-,H,r_S]).
\end{align*}

\item[-] \underline{Deviating} 
\begin{align*}
&p_{1,S} \left( \mu_L (c+1) \right) + (1-p_{1,S}) 2 \mu_H + \frac{\delta}{%
1-\delta} S_0.
\end{align*}
\end{itemize}

The incentive compatibility constraint can thus be written as 
\begin{align*}
p_{1,S} &\left( S_0 + \delta ((1-\gamma_L) u^*(c,L)+ \gamma_L \bar{v})
\right) + (1-p_{1,S}) (S_0 + \delta (\gamma_H u^*(1,L) + (1-\gamma_H) \bar{v}%
)) \\
&\leq p_{1,S} \left( \mu_L (c+1) \right) + (1-p_{1,S}) \left( 2 \mu_H
\right) + \frac{\delta}{1-\delta} S_0.
\end{align*}
where $u^*(c,L)$ and $u^*(1,L)$ are as defined in Lemma \ref{lemma:punish}.
It follows from Lemma \ref{lemma:punish}-1, \ref{lemma:punish}-3, and \ref%
{lemma:punish}-4 that 
\begin{equation*}
\delta \left( p_{1,S} ((1-\gamma_L) u^*(c,L) + \gamma_L \bar{v}) +
(1-p_{1,S})(\gamma_H u^*(1,L) + (1-\gamma_H) \bar{v}) \right) \leq \frac{%
\delta}{1-\delta} S_0.
\end{equation*}

The result is then proven by the following lemma.

\begin{lemma}
$S_0 \leq p_{1,S} (\mu_L (c+1)) + (1-p_{1,S})(2 \mu_H).$
\end{lemma}

\begin{proof}
Plugging in the expression of $p_{1,S}$ derived in Lemma \ref{lem:prob} and
rearranging 
\begin{align*}
\left( (1-\gamma_H) \frac{N-1}{N} + \gamma_H \frac{N-c}{N-1} \right) S_0
\leq \gamma_H \frac{N-c}{N-1} (\mu_L (c+1)) + (1-\gamma_H) \frac{N-1}{N} 2
\mu_H.
\end{align*}
Since $c+1 \geq \frac{S_0}{2 \mu_L} $ (by Lemma \ref{lemma:mlxl}) the
inequality above holds if 
\begin{align*}
&\left( (1-\gamma_H) \frac{N-1}{N} + \gamma_H \frac{N-c}{N-1} \right) S_0
\leq \gamma_H \frac{N-c}{N-1} \frac{S_0}{2} + (1-\gamma_H) \frac{N-1}{N} 2
\mu_H \\
\iff &\left((1-\gamma_H) \frac{N-1}{N} \right) S_0 + \gamma_H \frac{N-c}{N-1}
\frac{S_0}{2} \leq (1-\gamma_H) \frac{N-1}{N} 2 \mu_H.
\end{align*}
By Assumption \ref{ass:mul_muh} $S_0 \leq \mu_H$, so the result holds if 
\begin{equation*}
\gamma_H \frac{N-c}{N-1} \frac{1}{2} \leq (1-\gamma_H) \frac{N-1}{N}. 
\end{equation*}
The last inequality holds since, for all $N\ge2$, 
\begin{equation*}
\gamma_H \frac{N-c}{N-1} \frac{1}{2} < \frac{\gamma_H}{2} \le \frac{%
(1-\gamma_H)}{2} \le (1-\gamma_H)\frac{N-1}{N},
\end{equation*}
where we used $c>1$ and $\gamma_H \leq (1-\gamma_H)$ (since $\gamma_H\le
\frac12$).
\end{proof}
\end{itemize}

\subsubsection*{\textbf{Type 4 ($\protect\beta_t^i = U, r_t^i = r_R$)}}

These agents took safe at time $t-1$ and received a recommendation to take
risky, $r_R$. We show that each of these cases is comparable to one that has
already been detailed. Thus, following the recommendation is optimal. The
possible states of an agent of type 4 are:

\begin{itemize}
\item $\mathbf{[x_{t-1},\beta^i_{t-1},r^i_{t-1}]=[d, U, r_R]}$: receiving a
recommendation of risky means the agent is an experimenter. This is
equivalent to the case $[-, H, r_R]$ as this agent can infer that $%
\theta_{t-1} = H$ since, under $\pi_{c,d}$, the CP does not send a $r_R$ if
the flow is $d$ and $\theta = L$. Specifically, by Bayes rule 
\begin{align*}
&\mathbb{P}(\theta_{t-1} = L \mid [d, U, r_R]) \\[0.2cm]
&\quad = \frac{(1-\gamma_L)\mathbb{P}(r_R \mid \theta_{t-1} = L, x_{t-1} = d)%
}{(1-\gamma_L)\mathbb{P}(r_R \mid \theta_{t-1} = L, x_{t-1} = d) + \gamma_L 
\mathbb{P}(r_R \mid \theta_{t-1} = H, x_{t-1} = d) } = 0.
\end{align*}
Since it is a best response to take risky when the state is $[-, H, r_R]$,
it is also a best response to take risky when the state is $[d, U, r_R]$.

\item $\mathbf{[x_{t-1},\beta^i_{t-1},r^i_{t-1}]=[c, U, r_R]}$: receiving a
recommendation of risky could mean an agent is part of $d$ (if $\theta_{t-1}
= L$) \textbf{or} is the new experimenter (if $\theta_{t-1} = H$). Denote by 
$p_{c,R}$ agent's i belief that $\theta_{t-1} = L$. His expected cost of
following the recommendation is therefore 
\begin{align*}
u^*([c,U, r_R]) &= p_{c,R} u^*([d, L, r_R]) + (1-p_{c,R}) u^*([1, H, r_R]).
\end{align*}
The cost of deviating is the convex combination of deviating under the two
possible states to safe. In both of these cases the best response is to take
risky (see the states $[d, L, r_R]$ and $[1, H, r_R]$ discussed earlier).
Thus, it is a best response for the agent in this state to take risky.

\item $\mathbf{[x_{t-1},\beta^i_{t-1},r^i_{t-1}]=[1, U, r_R]}$: receiving a
recommendation of risky means that the agent is either a part of the
exploiter flow $c$ (if $\theta_{t-1}=L$) \textbf{or} has been selected to be
the next experimenter (if $\theta_{t-1}=H$). Denote by $p_{1,R}$ agent's i
belief that $\theta_{t-1} = L$. His expected cost of following the
recommendation and taking risky can then be written as 
\begin{align*}
u^*([1, U, r_R]) = p_{1,R} u^*([1, L, r_R]) + (1-p_{1,R}) u^*([1, H, r_R]).
\end{align*}
Similarly, the cost of deviating is a convex combination of deviating to
safe from state $[1, L, r_R]$ and deviating to safe from state $[1, H, r_R]$%
. In both of these cases the best response is to take risky. Thus, it is a
best response for the agent in this state to take risky.
\end{itemize}

\subsubsection*{\textbf{Type 5 (off the equilibrium path)}}

Off the equilibrium path the CP provides full information, that is, sends
recommendations of risky to each agent with probability $\frac{S_0}{%
\mu_{\beta} N}$. Note that in the full information regime every agent has the
same belief $\beta$ on the state of the risky road (equal to the belief of
the CP). Moreover the continuation cost for every agent is the same, no
matter his action. Since following the received recommendation is myopically
optimal no agent has an incentive to deviate. At each point every agent's
total expected cost is $1/(1-\delta) S_0$.

\subsection{Proof of Corollary \protect\ref{corollary:star_ic}}

The corollary follows from the following lemma.

\begin{lemma}
\label{lemma:xll} Let $\bar{x}_{LL}$ be the smallest integer $d$ that
satisfies \eqref{eq:def_xll_simple} when $c=x_{L}^{\textup{SO}}$, then 
\begin{equation}\label{xll_xeq}
\bar{x}_{LL} \le x^{\textup{eq}}_L.
\end{equation}
\end{lemma}
\begin{proof}
Note that 
\begin{align}
u^*([d, U, r_S])&=p_{d,S} \frac{S_0 + \delta \gamma_L \bar{v} }{%
1-\delta(1-\gamma_L)} + (1-p_{d,S}) \left( S_0 + \delta \left( \gamma_H
u^*(1,L) + (1-\gamma_H) \bar{v} \right) \right)  \notag \\
&=p_{d,S} \left(S_0 + \frac{\delta (S_0(1-\gamma_L)+ \gamma_L \bar{v}) }{%
1-\delta(1-\gamma_L)} \right) + (1-p_{d,S}) \left( S_0 + \delta \left(
\gamma_H u^*(1,L) + (1-\gamma_H) \bar{v} \right) \right). \notag
\end{align}
Hence by Lemma \ref{lemma:punish}
\begin{align}
u^*([d, U, r_S])\le p_{d,S} \left(S_0 + \frac{\delta S_0 }{1-\delta} \right)
+ (1-p_{d,S}) \left(S_0 + \frac{\delta S_0 }{1-\delta} \right)= S_0 + \frac{%
\delta S_0 }{1-\delta}.  \notag
\end{align}
A sufficient condition for \eqref{eq:def_xll_simple} to hold is therefore, 
\begin{equation*}
S_0 < p_{d,S} \mu_L (d+1)+(1-p_{d,S} )2\mu_H. 
\end{equation*}
We next show that $d=x_L^{\textup{eq}}$ satisfies this condition and thus %
\eqref{eq:def_xll_simple}. To this end, recall that $\mu_L (x_L^{\textup{eq}}+1)>S_0$ and $%
\mu_H \geq S_0$ hence 
\begin{equation*}
p_{d,S} \mu_L (d+1)+(1-p_{d,S} )2\mu_H > p_{d,S} S_0+(1-p_{d,S} )S_0=S_0, 
\end{equation*}
as desired.
\end{proof}

\subsection{Proof of Proposition \protect\ref{prop:ratio}}

\noindent The proposition follows from the following lemma.

\begin{lemma}
As $\delta \rightarrow 1$, $x_{LL} \rightarrow x_L^{\textup{SO}}$.
\end{lemma}

\begin{proof}
Recall $x_{LL}$ is the smallest integer $d$ (weakly greater than $x_L^{\textup{%
SO}}$) that satisfies \eqref{eq:def_xll_simple} for $c=x_L^{\textup{SO}}$. Specifically, plugging in values from Lemma \ref{lemma:closed}, taking the safe road (i.e.
following the recommendation) is the best response if 
\begin{align}
u^*([d, U, r_S]) &=p_{d,S} \frac{S_0 + \delta \gamma_L \bar{v} }{%
1-\delta(1-\gamma_L)} + (1-p_{d,S}) \left( S_0 + \delta \left( \gamma_H
u^*(1,L) + (1-\gamma_H) \bar{v} \right) \right)  \notag \\
&\leq p_{d,S} \mu_L (d+1) + (1-p_{d,S}) 2 \mu_H + \frac{\delta}{1-\delta}
S_0.  \label{eq:def_xll}
\end{align}
We let $c=d = x_L^{\textup{SO}}
$ in \eqref{eq:def_xll} and show as $\delta \rightarrow 1$ the constraint
holds. By Lemma \ref{lemma:punish}.3 and Lemma \ref{lemma:closed} it holds
\begin{align*}
u^*(1,L) &= \frac{x_L^{\textup{SO}}-1}{N-1} \left( \frac{1}{1-\delta(1-\gamma_L)} (\mu_L x_L^{\textup{SO}} +  \delta \gamma_L \bar{v}) \right) + \frac{N-x_L^{\textup{SO}}}{N-1} \left( \frac{1}{1-\delta(1-\gamma_L)} (S_0 + \delta \gamma_L \bar{v})\right) \\
&= \frac{1}{1-\delta(1-\gamma_L)} \left( \frac{x_L^{\textup{SO}}-1}{N-1} \mu_L x_L^{\textup{SO}} + \frac{N-x_L^{\textup{SO}}}{N-1} S_0 \right) + \bar{v} \left( \frac{\delta \gamma_L}{1-\delta(1-\gamma_L)} \right).
\end{align*}
Hence \eqref{eq:def_xll}  becomes 
\begin{equation} \label{eq11}
\begin{aligned} 
&\underbrace{p_{d,S} \left( \frac{1}{1-\delta(1-\gamma_L)} S_0 \right) + (1-p_{d,S})
\left( S_0 + \frac{\delta \gamma_H}{1-\delta(1-\gamma_L)} \left( \frac{x_L^{%
\textup{SO}}-1}{N-1} \mu_L x_L^{\textup{SO}} + \frac{N-x_L^{\textup{SO}}}{%
N-1} S_0 \right) \right)}_{T_1} \\
&+ \bar{v} \underbrace{\left( p_{d,S} \frac{\delta \gamma_L}{1-\delta(1-\gamma_L)} +
(1-p_{d,S}) \left( \frac{ \delta^2 \gamma_H \gamma_L}{1-\delta(1-\gamma_L)}  +
\delta (1-\gamma_H)\right) \right)}_{T_2} \\
&\leq \underbrace{p_{d,S} \left( \mu_L (x_L^{\textup{SO}} +1) \right) + (1-p_{d,S})
\left( 2 \mu_H \right)}_{T_3} + \underbrace{\frac{\delta}{1-\delta} S_0}_{T_4}.
\end{aligned}
\end{equation}
It follows from \eqref{v_bar} with $g(c,\mu_L)=g(d,\mu_L)=g(x_L^{\textup{SO}},\mu_L)$ that
\begin{align*}
\bar{v} = \frac{1}{(1-\delta)} \underbrace{ \frac{1-\delta(1-\gamma_L)}{N(1-\delta(1-\gamma_H-%
\gamma_L)} \left( \mu_H + (N-1) S_0 + \frac{\delta
\gamma_H}{1-\delta(1-\gamma_L)} \left(  \mu_L
(x_L^{\textup{SO}} )^2+(N-x_L^{\textup{SO}}) S_0 \right) \right)}_{T_5}.
\end{align*}
Substituting in \eqref{eq11} and multiplying both sides by $(1-\delta)$ yields 
$$ T_1(1-\delta) + T_5 T_2 \le T_3 (1-\delta)+\delta S_0.$$
Since $T_1$ and $T_3$ are finite and $\lim_{\delta\rightarrow 1}T_2= 1$,
 when $\delta
\rightarrow 1$ a sufficient condition for \eqref{eq11} to hold is 
\begin{align*}
\lim_{\delta \rightarrow 1} T_5 := &\frac{\gamma_L}{(\gamma_H+ \gamma_L)N} \left(\mu_H + (N-1)
S_0 + \frac{\gamma_H}{\gamma_L} \left(  \mu_L
(x_L^{\textup{SO}})^2 + (N-x_L^{\textup{SO}}) S_0 \right) \right) <
S_0 \\
\iff & \gamma_L \mu_H + \gamma_L S_0 (N-1) + \gamma_H (\mu_L (x_L^{\textup{%
SO}})^2 + S_0 (N-x_L^{\textup{SO}})) < (\gamma_H + \gamma_L) S_0 N \\
\iff & \mu_H < S_0 + \frac{\gamma_H}{\gamma_L} (  - \mu_L (x_L^{%
\textup{SO}})^2+S_0x_L^{\textup{SO}}) \\
&\quad = S_0 + \frac{\gamma_H}{\gamma_L} x_L^{\textup{SO}} (S_0 - \mu_L
x_L^{\textup{SO}}).
\end{align*}
The inequality holds by Assumption \ref{ass:mul_muh} if 
\begin{equation*}
\frac{\gamma_H}{\gamma_L} x_L^{\textup{SO}} (S_0 - \mu_L x_L^{\textup{SO}}) >
\gamma_H \left(\frac{S_0}{3} - \mu_L \right)
\end{equation*}
this follows from $1/ \gamma_L > 1$ and 
\begin{align*}
x_L^{\textup{SO}} (S_0 - \mu_L x_L^{\textup{SO}}) \geq x_L^{\textup{SO}}
\left( S_0 - \frac{S_0}{2} - \frac{\mu_L}{2} \right) = \frac{x_L^{\textup{SO%
}}}{2} (S_0 - \mu_L) \geq S_0 - \mu_L \ge \frac{S_0}{3} - \mu_L
\end{align*}
where the first inequality follows from $x_L^{\textup{SO}} \leq \frac{S_0}{%
2 \mu_L} + \frac{1}{2}$ (can be proven similarly as in Lemma~\ref{lemma:closed}) and the second inequality follows from $x_L^{\textup{SO%
}} \geq 2$.
\end{proof}

\subsection{Proof of Proposition \protect\ref{prop:delta_small}}

We have $\pi^*=\pi_{1,1,x_{L}^{\textup{SO}},x_{LL}}$ and $%
V^*=V_{1,1,x_{L}^{\textup{SO}},x_{LL}}$. We denote an arbitrary optimal incentive
compatible scheme (OICS) in $\hat \Pi$ as $\bar \pi:=\pi_{\bar a,\bar b,\bar
c,\bar d}$ with social cost $\bar V$. We proceed in steps.

\begin{enumerate}
\item \textbf{Proof of Lemma \ref{lem:small}: Any OICS $\boldsymbol{\bar \pi}$ must satisfy} $\boldsymbol{%
\bar a,\bar b \leq 1}$\newline
For any recommendation in $\hat{\Pi}$ the experimenters, i.e. the agents
chosen to take risky at time $t$ when $\theta_{t-1} = H$ are selected at random from all
agents. Thus, with positive probability the agent or agents chosen as part
of $a$ or $b$ know that $\theta_{t-1} = H$. We next show that if $a$ or $b$ is greater than one,
then it is not incentive compatible for an agent to follow this
recommendation when $\delta \leq 1/2$. Specifically, following gives cost of
at least 
\begin{equation*}
2 \mu_H + \delta ((1-\gamma_H) \hat{v}_L + \gamma_H \hat{v}_H) 
\end{equation*}
where $\hat{v}_L $ and $\hat{v}_H$ are some positive continuation costs.
Deviating to safe has cost 
\begin{equation*}
\frac{1}{1-\delta} S_0 \le 2S_0 \quad \textup{for} \quad \delta \leq 1/2. 
\end{equation*}
Since $\mu_H \geq S_0$ and the continuation costs $\hat{v}_L $ and $\hat{v}_H
$ are positive it is not incentive compatible for the agent to follow and
take risky when $a,b\ge2$. Thus, for incentive compatibility to hold at most
one agent can be sent to risky.


\item \textbf{Any OICS satisfies $\bar c=x^{\textup{SO}}_L$ or $\bar c=x^{\textup{SO}}_L+1$.} {%
In this proof we assume $x_{LL}>x^{\textup{SO}}_L$ (if not $\pi^*$ coincides with the
social optimum and is thus already an OICS) hence $x_{LL}$ is the minimum
integer satisfying \eqref{eq:def_xll} for $c=x_{L}^{\textup{SO}}$.} The incentive
compatibility constraint \eqref{eq:def_xll} can be equivalently rewritten as 
\begin{equation}  \label{IC_cd}
\begin{aligned} &p_{d,S} u^*([d,L,r_S]) + (1-p_{d,S}) \left( S_0 + \delta
\left( (1-\gamma_H)\bar{v} + \gamma_H \left( \frac{c-1}{N-1} \mu_L c +
\frac{N-c}{N-1} S_0 \right. \right. \right. \\ &\qquad \qquad \left. \left.
\left. + \delta \left( \gamma_L \bar{v} + (1-\gamma_L) \left(
\frac{d-1}{N-1} u^*([d,L,r_R]) + \frac{N-d}{N-1} u^*([d,L,r_S]) \right)
\right) \right) \right) \right) \\ &\leq p_{d,S} \mu_L (d+1) + (1-p_{d,S}) 2
\mu_H + \frac{\delta}{1-\delta} S_0. \end{aligned}
\end{equation}
Recall from Lemma \ref{lem:prob} that
\begin{align} \label{eq:pds}
p_{d,S}=\frac{1-\gamma_L}{1-\gamma_L+ \gamma_L \frac{N-1}{N}}
\end{align}
does not depend on $d$, hence within this proof to avoid confusion we denote this by $p_S$.
By substituting the expressions of $u^*([d,L,r_S])$ and $u^*([d,L,r_R])$
computed in Lemma \ref{lemma:closed} and the expression of $\bar{v}$ given
in \eqref{v_bar} the LHS of the IC constraint \eqref{IC_cd} can be rewritten
as 
\begin{align*}
&p_{S} u^*([d,L,r_S]) + (1-p_{S}) \left( S_0 + \delta \left( (1-\gamma_H)%
\bar{v} + \gamma_H \left( \frac{c-1}{N-1} \mu_L c + \frac{N-c}{N-1} S_0
\right. \right. \right. \\
&\qquad \qquad \left. \left. \left. + \delta \left( \gamma_L \bar{v} + \frac{%
(1-\gamma_L)}{1-\delta(1-\gamma_L)} \left( \frac{d-1}{N-1} \mu_L d + \frac{%
N-d}{N-1} S_0 + \delta \gamma_L \bar{v} \right) \right) \right) \right)
\right)  \notag \\
&= p_{S} \frac{1}{1-\delta(1-\gamma_L)} \left( S_0 + \delta \gamma_L \bar{v%
} \right) + (1-p_{S}) \left( S_0 + \delta \left( (1-\gamma_H)\bar{v} + \gamma_H \left( \frac{c-1}{N-1} \mu_L c \right. \right. \right. \\
& \qquad \qquad \left. \left. \left. + \frac{N-c}{N-1} S_0 + \frac{\delta}{1-\delta(1-\gamma_L)}
\left( \gamma_L \bar{v} + (1-\gamma_L) \left( \frac{d-1}{N-1} \mu_L d + 
\frac{N-d}{N-1} S_0 \right) \right) \right) \right) \right) \\
&= \frac{p_{S} S_0}{1-\delta(1-\gamma_L)} + (1-p_{S}) \left( S_0 + \delta \gamma_H \left( \frac{c-1}{N-1}
\mu_L c + \frac{N-c}{N-1} S_0 \right. \right. \\
&\qquad \qquad \left. \left. + \frac{\delta (1-\gamma_L)}{1-\delta(1-\gamma_L)} \left( \frac{d-1}{N-1} \mu_L d + \frac{N-d}{N-1} S_0
\right) \right) \right) \\
&\qquad \qquad + \bar{v} \left( \underbrace{p_{S} \frac{\delta \gamma_L}{%
1-\delta(1-\gamma_L)} + (1-p_{S}) \delta \left((1-\gamma_H) + \gamma_H 
\frac{\delta \gamma_L}{1-\delta(1-\gamma_L)} \right)}_{\tau} \right) \\
\end{align*}
\begin{align*}
&=\frac{p_{S} S_0}{1-\delta(1-\gamma_L)} + (1-p_{S}) \left( S_0 + \delta \gamma_H \left( \frac{c-1}{N-1}
\mu_L c + \frac{N-c}{N-1} S_0 +\right. \right. \\
& \qquad \qquad \left. \left. \frac{\delta (1-\gamma_L)}{1-\delta(1-\gamma_L)} \left( \frac{d-1}{N-1} \mu_L d + \frac{N-d}{N-1} S_0
\right) \right) \right) \\
&\qquad \qquad + {\frac{\tau \tilde \tau}{N}\left[ g(1,\mu_H)
+\delta\gamma_H g(c,\mu_L)+\delta^2 \frac{(1-\gamma_L)}{1-\delta(1-\gamma_L)}%
\gamma_H g(d,\mu_L) \right]},
\end{align*}
for 
\begin{equation}\label{tau}
\begin{aligned}
\tilde{\tau}:=&\frac{1-\delta(1-\gamma_L)}{(1-\delta)(1-\delta(1-\gamma_H -\gamma_L))}\\
\tau :=& p_{S} \frac{\delta \gamma_L}{1-\delta(1-\gamma_L)} +
(1-p_{S}) \delta \left((1-\gamma_H) + \gamma_H \frac{\delta \gamma_L}{%
1-\delta(1-\gamma_L)} \right) 
\end{aligned}
\end{equation}
independent of $a,b,c,d$. Note that this is separable in $c$ and $d$. {%
Specifically, by bringing all the terms depending on c on the LHS and all
the terms depending on $d$ on the RHS, we can rewrite the IC constraint %
\eqref{IC_cd} as $f(c)\le g(d)$ with} 
\begin{equation}  \label{IC_cd2}
\begin{aligned} f(c)&:={ \frac{\delta(1-p_{S}) \gamma_H}{N-1}
\underbrace{\left( (c-1) \mu_L c + (N-c) S_0\right)}_{=:f_1(c)}
+\frac{\tau \tilde \tau}{N} \delta\gamma_H \underbrace{g(c,\mu_L)}_{=:f_2(c)} +k}\\
g(d)&:=p_{S} \mu_L (d+1) - \delta(1-p_{S}) \gamma_H \frac{\delta
(1-\gamma_L)}{1-\delta(1-\gamma_L)} \left( \frac{d-1}{N-1} \mu_L d +
\frac{N-d}{N-1} S_0 \right) \\ &\qquad \qquad - \frac{\tau \tilde \tau}{N} \delta^2
\frac{(1-\gamma_L)}{1-\delta(1-\gamma_L)}\gamma_H g(d,\mu_L)\\
k&:={\frac{p_{S} S_0}{1-\delta(1-\gamma_L)} + (1-p_{S}) S_0}+
\frac{\tau \tilde \tau}{N} g(1,\mu_H)-(1-p_{S}) 2 \mu_H -\frac{\delta}{1-\delta} S_0.
\end{aligned}
\end{equation}
Note that $k$ is a constant independent of $c$ and $d$. The functions $f_1(c)
$ and $f_2(c)$ are quadratic, convex and, disregarding the integer
constraint, are minimized when $c = \frac{\mu_L + S_0}{2 \mu_L}= \frac{S_0}{%
2 \mu_L} + \frac{1}{2}$ and when $c = \frac{S_0}{2 \mu_L}$, respectively.
This means that, disregarding the integer constraint, $f(c)$ is minimized for some $\tilde c^*\in \left[ \frac{S_0}{2
\mu_L}, \frac{S_0}{2 \mu_L} + \frac{1}{2} \right].$ Let $c^*$ be the integer
minimizer of $f(c)$ (i.e. the integer closest to $\tilde c^*$). Recall that 
$x_{L}^{\textup{SO}}$ is the integer minimizer of $g(c,\mu_L)$ (i.e. the integer
closest to $\frac{S_0}{2 \mu_L}$).

There are two possible cases:

\begin{enumerate}
\item \textit{if there exists an integer $\hat c$ in the interval $\left[ 
\frac{S_0}{2 \mu_L}, \frac{S_0}{2 \mu_L} + \frac{1}{2} \right]$ then $%
c^*=\hat c=x_{L}^{\textup{SO}}$.}\newline
In fact, $|\hat c-\frac{S_0}{2 \mu_L}|\le \frac12 \Rightarrow \hat
c=x_{L}^{\textup{SO}} $ and $|\hat c-\tilde c^*|\le \frac12 \Rightarrow \hat c=c^*$;

\item \textit{if there is no integer in the interval $\left[ \frac{S_0}{2
\mu_L}, \frac{S_0}{2 \mu_L} + \frac{1}{2} \right]$ then either $%
c^*=x_{L}^{\textup{SO}}$ or $c^*=x_{L}^{\textup{SO}}+1$.}\newline
In fact, let $\hat c^-$ be the largest integer smaller than $\frac{S_0}{2
\mu_L}$ and $\hat c^+$ be the smallest integer larger than $\frac{S_0}{2
\mu_L}+\frac12$ (i.e. $\hat c^-+1$). Since there is no integer in $\left[ 
\frac{S_0}{2 \mu_L}, \frac{S_0}{2 \mu_L} + \frac{1}{2} \right]$ it must be $%
|\hat c^- -\frac{S_0}{2 \mu_L}|\le \frac12 \Rightarrow \hat c^-=x_{L}^{\textup{SO}}$.
On the other hand $c^*$ is equal to either $\hat c^-$ or $\hat c^+$
depending on whether $\tilde c^*$ is smaller or larger than $\frac{\hat
c^-+\hat c^+}{2}$.
\end{enumerate}

Thus, $f(c)$ is minimized at $c^* =
x_L^{\textup{SO}}$ or $c^* = x_L^{\textup{SO}}+1$ and for any $\tilde
c\neq \{x_L^{\textup{SO}}, x_L^{\textup{SO}}+1\}$, we get $f(\tilde c)\ge
f(x_L^{\textup{SO}})$.

We now prove that if $(\tilde c,\tilde d)$ satisfies \eqref{eq:def_xll} with 
$c\neq \{x_L^{\textup{SO}}, x_L^{\textup{SO}}+1\}$ then $\tilde d \ge
x_{LL}$. In fact it must be $g(\tilde d) \ge f(\tilde c)\ge f(x_L^{\textup{SO}})$.
Since $x_{LL}$ is the minimum integer satisfying $g(d)\ge f(x_L^{\textup{SO}}) $ it
must be $\tilde d \ge x_{LL}$.

Hence any IC scheme with $\tilde{c} \neq \{x_L^{\textup{SO}}, x_L^{\textup{SO}}+1\}$
has higher cost then $\pi^*$ (since $\tilde{c} \geq x_L^{\textup{SO}}$ and $\tilde{d} \geq x_{LL}$) and cannot be optimal.\footnote{Recall that by \eqref{v_bar} the cost is $\tilde{\tau} \left[ g(1,\mu_H)
+\delta\gamma_H g(c,\mu_L)+\delta^2 \frac{(1-\gamma_L)}{1-\delta(1-\gamma_L)} \gamma_H g(d,\mu_L) \right]$ and $g(c, \mu_L)/g(d,\mu_L)$ are minimized for $c=x_L^{\textup{SO}}/d=x_L^{\textup{SO}}$ and strictly increasing functions for larger values of $c/d$.}

\item \textbf{Any OICS satisfies $\bar d\ge x^{LL}-1$ }

From the previous point we know, that in any OICS it must be $\bar c\in
\{x_L^{\textup{SO}}, x_L^{\textup{SO}}+1\}$.
\begin{itemize}
\item If $\bar c = x_L^{\textup{SO}}$ then by definition $x_{LL}$ is the
minimum value of $\bar d$ to maintain incentive compatibility.

\item If instead $\bar c = x_L^{\textup{SO}} + 1$ we show that to maintain
incentive compatibility it must be $\bar d \geq x_{LL}-1$. In fact suppose
by contradiction that $f( x_L^{\textup{SO}} + 1)\le g(x_{LL}-2)$, then we
show that under our assumptions $f( x_L^{\textup{SO}})\le g(x_{LL}-1)$,
which is absurd since $x_{LL}$ is the minimum integer satisfying the IC
constraint. To this end it suffices to show  
\begin{equation*}
f( x_L^{\textup{SO}})-f( x_L^{\textup{SO}} + 1) \le
g(x_{LL}-1)-g(x_{LL}-2). 
\end{equation*}
Note that
\begin{equation} \label{step}
\begin{aligned}
&f( x_L^{\textup{SO}})-f( x_L^{\textup{SO}} + 1) \le
g(x_{LL}-1)-g(x_{LL}-2) \\
\iff &\frac{\delta (1-p_{S}) \gamma_H}{N-1} \left( S_0 - 2 x_L^{%
\textup{SO}} \mu_L + \frac{\delta (1-\gamma_L)}{1-\delta(1-\gamma_L)} (S_0
- 2(x_{LL}-2) \mu_L) \right) \\
&+ \frac{\tau \tilde{\tau}}{N} \delta \gamma_H \left( S_0 - (2 x_L^{\textup{%
SO}} + 1) \mu_L + \frac{\delta (1-\gamma_L)}{1-\delta(1-\gamma_L)} (S_0 -
(2(x_{LL}-2)+1)\mu_L) \right) \\
& \qquad \leq p_{S} \mu_L.
\end{aligned}
\end{equation}

Now, by $x_{LL} \geq x_L^{\textup{SO}} + 1$\footnote{Recall that if $x_{LL} = x_L^{\textup{SO}}$ then $\pi^*$ is equivalent to the social optimum and it is therefore the OICS.} and $x_L^{\textup{SO}} \geq 
\frac{S_0}{2 \mu_L} - \frac{1}{2}$ (see Lemma \ref{lemma:mlxl}) the following two inequalities hold 
\begin{align*}
S_0 -
(2(x_{LL}-2)+1)\mu_L& \leq S_0 -
2(x_{LL}-2)\mu_L \leq S_0 - 2
\mu_L (x_L^{\textup{SO}}-1) \\
& \leq S_0 - S_0 + \mu_L +2\mu_L= 3\mu_L \\
 S_0 - (2x_L^{\textup{SO}}+1) \mu_L &\leq S_0 - 2 \mu_L x_L^{\textup{SO}}
\leq \mu_L
\end{align*}
and by $\delta \leq 1/2$ 
\begin{equation*}
\frac{\delta(1-\gamma_L)}{1-\delta(1-\gamma_L)} \leq 1.
\end{equation*}
Thus, the left hand side of \eqref{step} can be upper bounded by 
\begin{align*}
4\mu_L\left(\frac{\delta(1-p_{S}) \gamma_H}{N-1}  + \frac{\tau \tilde{\tau}%
}{N} \delta \gamma_H \right).
\end{align*}
So it is sufficient if 
\begin{align}
&4 \left( \frac{\delta (1-p_{S})\gamma_H}{N-1} + \frac{\tau \tilde{%
\tau}}{N} \delta \gamma_H \right) \leq p_{S} \nonumber \\
\iff &4 \delta \gamma_H \left( (1-p_{S}) + \frac{N-1}{N} \tau \tilde{%
\tau} \right) \leq (N-1) p_{S} \nonumber \\
\overset{\eqref{tau}}{\iff} &4 \delta \gamma_H \left( (1-p_{S}) + \frac{N-1}{N} \tilde{\tau}
\left( p_{S} \frac{\delta \gamma_L}{1-\delta(1-\gamma_L)} + \right. \right. \nonumber \\
& \qquad \qquad \left. \left. (1-p_{S}) \delta \left( (1-\gamma_H) + \gamma_H \frac{\delta \gamma_L
}{1-\delta(1-\gamma_L)} \right) \right) \right) \leq (N-1) p_{S} \nonumber \\
\overset{\eqref{eq:pds}}{\iff} & 4 \delta \gamma_H \left(\gamma_L \frac{N-1}{N} + \frac{N-1}{N} \tilde{%
\tau} \left( \frac{\delta \gamma_L (1-\gamma_L)}{1-\delta(1-\gamma_L)} \right. \right. \nonumber \\ &
\left. \left. \qquad \qquad + \gamma_L \frac{N-1}{N} \delta \left((1-\gamma_H) + \gamma_H \frac{\delta
\gamma_L }{1-\delta(1-\gamma_L)} \right)\right) \right) \leq (N-1) (1-\gamma_L) \nonumber \\
\iff &4 \delta \gamma_H \left( \gamma_L + \tilde{\tau} \underbrace{\left( \frac{\delta
\gamma_L (1-\gamma_L)}{1-\delta(1-\gamma_L)} + \gamma_L \frac{N-1}{N} \delta
\left((1-\gamma_H) + \gamma_H \frac{\delta \gamma_L }{%
1-\delta(1-\gamma_L)} \right)\right)}_{\hat \tau} \right)\nonumber \\
&\qquad \qquad \leq (1-\gamma_L) N \nonumber \\
\iff & 4 \delta \gamma_H \left( \gamma_L + \tilde{\tau} \hat\tau\right) \le (1-\gamma_L)N. \label{formula1}
\end{align}
where the second equivalence comes from plugging in $\tau$ and the third from plugging in $p_{S}$ and multiplying by $\left( 1-\gamma_L + \gamma_L \frac{N-1}{N} \right)$.

Now, we show $\tilde{\tau} \leq 2$. By \eqref{tau}\begin{align*}
&\tilde{\tau}=\frac{1-\delta(1-\gamma_L)}{(1-\delta)(1-\delta(1-\gamma_H -\gamma_L))} =\frac{1-\delta(1-\gamma_L)}{(1-\delta(1-\gamma_H))(1-\delta(1-\gamma_L)) -
\delta^2 \gamma_H \gamma_L} \leq 2 \\
&\iff  1-\delta(1-\gamma_L) \leq 2 (
(1-\delta(1-\gamma_H))(1-\delta(1-\gamma_L)) - \delta^2 \gamma_H \gamma_L) \\
&\iff  1 \leq 2 \left( (1-\delta(1-\gamma_H)) - \frac{\delta^2 \gamma_H
\gamma_L}{1-\delta(1-\gamma_L)} \right) \\
&\iff  1 \leq 2 \left(1-\delta + \delta \gamma_H \left( 1- \frac{\delta
\gamma_L}{1-\delta(1-\gamma_L)} \right) \right).
\end{align*}
By $\frac{\delta \gamma_L}{1-\delta(1-\gamma_L)} \leq 1$ the right hand side
can be lower bounded by $2(1-\delta)$. Then by $\delta \leq 1/2$ the
inequality holds. Now we bound $\hat{\tau}$ by $1$ 
\begin{align*}
&\frac{\delta \gamma_L (1-\gamma_L)}{1-\delta(1-\gamma_L)} + \gamma_L \frac{%
N-1}{N} \delta \left((1-\gamma_H) + \gamma_H \frac{\delta \gamma_L }{%
1-\delta(1-\gamma_L)} \right) \\
&= \frac{\delta\gamma_L(1-\gamma_L)+ \delta \gamma_L \frac{N-1}{N}
((1-\gamma_H)(1-\delta(1-\gamma_L))+\delta\gamma_H \gamma_L)}{1-\delta(1-\gamma_L)}
\\
&\leq \frac{\delta(1-\gamma_L)+ \delta \gamma_L}{1-\delta(1-\gamma_L)} = 
\frac{\delta}{1-\delta(1-\gamma_L)} \leq 1.
\end{align*}

Plugging in the bounds for $\tilde{\tau}$ and $\hat{\tau}$ in \eqref{formula1} leads to the sufficient condition
\begin{align*}
4 \delta \gamma_H (\gamma_L + 2) \leq (1-\gamma_L) N.
\end{align*}
Then by $\delta \leq 1/2, \gamma_H \leq 1/2, \gamma_L \leq 1/2$ it is
sufficient if $N \geq 5$, which is true by assumption.
\end{itemize}
\end{enumerate}

We are now ready to prove the two main statements:\newline

\textbf{Proof of statement 2:}\newline
Overall, we know that $\pi^*$ is incentive compatible and achieves minimum
cost among the IC schemes with $\bar c=x_{L}^{\textup{SO}}$. From the points above,
we know that the only other possibility is $\bar c=x_{L}^{\textup{SO}}+1$ in which
case $\bar d\ge x_{LL}-1$. Any choice of $\bar d\ge x_{LL}$ leads to higher
cost than $V^*$ hence the only possibility left is $\tilde{\pi}^*$. If $%
\tilde{\pi}^*$ is IC and has cost $\tilde{V}^*$ lower than $V^*$ then that
is the OICS otherwise $\pi^*$ is. Whether that happens or not depends on the
chosen parameters. \newline

\textbf{Proof of statement 1:}\newline
We next show that for $\delta\rightarrow 0$, $\tilde{V}^*>V^*$. To this end,
recall the expression for $\bar v$ in \eqref{v_bar} and $\tilde \tau$ in \eqref{tau}. Then 
\begin{align*}
V^*&:= \tilde \tau \left[ g(1,\mu_H) +\delta\gamma_H
g(x_{L}^{\textup{SO}},\mu_L)+\delta^2 \frac{(1-\gamma_L)}{1-\delta(1-\gamma_L)}%
\gamma_H g(x_{LL},\mu_L) \right], \\
\tilde{V}^*&:= \tilde \tau \left[ g(1,\mu_H) +\delta\gamma_H
g(x_{L}^{\textup{SO}}+1,\mu_L)+\delta^2 \frac{(1-\gamma_L)}{1-\delta(1-\gamma_L)}%
\gamma_H g(x_{LL}-1,\mu_L) \right],
\end{align*}
For $\delta\rightarrow 0$ the terms in $\delta^2$ are negligible and the
conclusion follows by $g(x_{L}^{\textup{SO}},\mu_L)<g(x_{L}^{\textup{SO}}+1,\mu_L).$

\bibliographystyle{plainnat}
\bibliography{routing_bib}

\end{document}